\documentclass[11pt]{article}
\usepackage[T1]{fontenc}
\usepackage[utf8]{inputenc}
\usepackage[
letterpaper,
margin=1in
]{geometry}

\usepackage{setspace}
\usepackage{multirow}
\usepackage{graphicx}
\usepackage{dsfont}
\usepackage{mathtools}
\usepackage{amssymb}
\usepackage{amsthm}
\usepackage{thmtools}
\usepackage{amsfonts}
\usepackage{amsmath}
\usepackage{accents}
\usepackage{bm}
\usepackage{float}
\usepackage[normalem]{ulem}
\usepackage{xcolor}
\definecolor{citegreen}{HTML}{208054}
\definecolor{citeblue}{HTML}{0055cc}
\usepackage[
colorlinks=true,
linkcolor=citeblue,
citecolor=citegreen,
linktocpage=true
]{hyperref}
\usepackage[capitalise]{cleveref}
\usepackage{comment}
\usepackage{orcidlink}

\usepackage{physics}
\usepackage[utf8]{inputenc}

\usepackage{graphicx}
\usepackage{subcaption}

\newtheorem{theorem}{Theorem}[section]
\newtheorem*{theorem*}{Theorem}
\newtheorem{definition}[theorem]{Definition}
\newtheorem{remark}[theorem]{Remark}
\newtheorem{fact}[theorem]{Fact}
\newtheorem{proposition}[theorem]{Proposition}
\newtheorem{lemma}[theorem]{Lemma}
\newtheorem{claim}[theorem]{Claim}
\newtheorem{corollary}[theorem]{Corollary}
\newtheorem*{corollary*}{Corollary}

\usepackage{tikz}
\usetikzlibrary{tikzmark}
\usetikzlibrary{shapes,snakes, arrows.meta}
\usepackage{dirtytalk}

\usepackage{pgfplots}
\pgfplotsset{compat=1.13}

\usepackage[style=alphabetic,backref=true,maxnames=99,maxalphanames=4,natbib=true,isbn=false]{biblatex}
\renewbibmacro*{doi+eprint+url}{%
  \iftoggle{bbx:doi}
    {\iffieldundef{url}{\printfield{doi}}{}}
    {}%
  \newunit\newblock
  \iftoggle{bbx:eprint}
    {\usebibmacro{eprint}}
    {}%
  \newunit\newblock
  \iftoggle{bbx:url}
    {\usebibmacro{url+urldate}}
    {}%
}
\DeclareSourcemap{
  \maps[datatype=bibtex]{
    \map{
      \step[fieldset=urldate, null]
    }
  }
}

\DefineBibliographyStrings{english}{
backrefpage = {p.},
backrefpages = {pp.}
}
\def\be{\begin{equation}}
\def\ee{\end{equation}}
\def\bea{\begin{eqnarray}}
\def\eea{\end{eqnarray}}
\def\bean{\begin{eqnarray*}}
\def\eean{\end{eqnarray*}}
\def\gsim{\mathrel{\rlap{\lower0.2em\hbox{$\sim$}}\raise0.2em\hbox{$>$}}}
\def\ksim{\mathrel{\rlap{\lower0.2em\hbox{$\sim$}}\raise0.2em\hbox{$<$}}}
\def\kg{\mathrel{\rlap{\lower0.25em\hbox{$>$}}\raise0.25em\hbox{$<$}}}

\newcommand\glnf[1]{$\mathrm{GL}_n\left(\mathbb{F}_2\right)$}
\newcommand{\pmone}{\{\pm 1\}}

\renewcommand{\O}{\mathcal{O}}

\newcommand{\R}{\mathbb{R}}
\newcommand{\C}{\mathbb{C}}

\newcommand{\I}{\mathbb{I}}

\renewcommand{\i}{\mathrm{i}}
\newcommand{\Orth}{\mathrm{O}}
\newcommand{\SO}{\mathrm{SO}}
\newcommand{\SU}{\mathrm{SU}}
\newcommand{\U}{\mathrm{U}}

\let\tr\relax
\DeclareMathOperator{\tr}{tr}

\DeclareMathOperator{\pf}{pf}

\DeclareMathOperator{\LOCC}{LOCC}
\DeclareMathOperator{\LCOCC}{LCOCC}
\renewcommand{\op}[2]{\ket{#1}\!\bra{#2}}
\renewcommand{\ip}[2]{\langle #1 | #2 \rangle}
\newcommand{\melem}[3]{\langle #1 | #2 | #3 \rangle}
\renewcommand{\ev}[2]{\melem{#2}{#1}{#2}}
\renewcommand{\ket}[1]{| #1 \rangle}
\renewcommand{\bra}[1]{\langle #1 |}

\newcommand{\T}{\mathrm{T}}

\renewcommand{\l}[1]{\mathopen{}\left#1}
\renewcommand{\r}[1]{\right#1\mathclose{}}

\newcommand{\BPP}{\mathsf{BPP}}
\newcommand{\BQP}{\mathsf{BQP}}
\newcommand{\sharpP}{\#\mathsf{P}}
\newcommand{\Family}{\mathcal{F}}
\newcommand{\Hadamard}{\mathrm{H}}
\newcommand{\Phase}{\mathrm{S}}
\newcommand{\CZ}{\mathrm{CZ}}
\newcommand{\subspace}[1]{\mathcal{H}(#1)}
\newcommand{\VS}{\mathrm{VS}}
\newcommand{\HS}{\mathrm{HS}}
\newcommand{\edgeops}{\mathcal{E}}

\AtBeginDocument{%
  \typeout{*** catcode of &: \the\catcode`\&}%
  \typeout{*** meaning of &: \meaning\&}%
}

\begin{document}

\title{Fermionic Insights into Measurement-Based Quantum Computation: Circle Graph States Are Not Universal Resources}
\author{
Brent Harrison\thanks{Dartmouth College} \,\orcidlink{0000-0003-4056-6888}
\and
Vishnu Iyer\thanks{University of Texas at Austin} \,\orcidlink{0000-0001-8072-1390}
\and
Ojas Parekh\thanks{Sandia National Laboratories} \,\orcidlink{0000-0003-2689-9264}
\and
Kevin Thompson\thanks{Sandia National Laboratories} \,\orcidlink{0000-0001-5669-2200}
\and
Andrew Zhao\thanks{Sandia National Laboratories} \,\orcidlink{0000-0002-0299-0277}
}

\date{May 15, 2026}

\maketitle

\begin{abstract}
Measurement-based quantum computation (MBQC) is a strong contender for realizing quantum computers. A critical question for MBQC is the identification of resource graph states that can enable universal quantum computation. Any such universal family must have unbounded entanglement width, which is equivalent to the ability to produce any circle graph state from the states in the family using only local Clifford operations, local Pauli measurements, and classical communication. Yet, it was not previously known whether or not circle graph states themselves are a universal resource. We show that, in spite of their expressivity, circle graph states are not efficiently universal for MBQC (i.e., assuming $\mathsf{BQP} \neq \mathsf{BPP}$). We prove this by articulating a precise graph-theoretic correspondence between circle graph states and a certain subset of fermionic Gaussian states. This is accomplished by synthesizing a variety of techniques that allow us to handle both stabilizer states and fermionic Gaussian states at the same time. As such, we anticipate that our developments may have broader applications beyond the domain of MBQC as well.
\end{abstract}

\clearpage

\tableofcontents

\clearpage

\section{Introduction}

Measurement-based quantum computation (MBQC) is a promising route to develop practical, large-scale quantum computers \cite{raussendorf_one-way_2001,perdrix_measurement-based_2004,jozsa_introduction_2005,briegel2009measurement,bourassa2021blueprint}. At a high level, MBQC involves the synthesis of so-called resource states such that an appropriately chosen sequence of single-qubit measurements on these states simulates some quantum circuit. A family of resource states is said to be universal if MBQC can simulate \emph{any} quantum computation in this manner. The MBQC paradigm therefore shifts the bulk of the computational effort onto the preparation of such resource states. This approach is enticing because it allows for the possibility of offline resource-state preparation, thereby reducing the online work to the relatively straightforward task of performing (adaptive) single-qubit measurements. Considerable effort has been made toward identifying which resource states allow for universal MBQC \cite{van_den_nest_universal_2006,bremner2009random,gross2009most,daniel2020computational}.

Typically, candidates for resource states are sought within the set of graph states \cite{hein2006entanglement}. These are stabilizer states uniquely specified by a graph, and this structure makes their preparation conceptually straightforward, requiring only controlled-$Z$ gates applied to $\ket{+}$ states. Furthermore, every stabilizer state is local Clifford equivalent (LC-equivalent) to some graph state and vice versa \cite{schlingemann2002stabilizer,hein2004multiparty,van2004graphical}. Thus it is typical to restrict attention to graph states when searching for universal resource states.

Circle graphs constitute a family of graphs that is particularly relevant to this topic. A circle graph is defined by a chord diagram on a circle:~each chord represents a vertex, and two vertices are adjacent if and only if their corresponding chords intersect. Circle graph states (CGSs) appear naturally in MBQC as follows. Two graph states are reachable from each other, using only local Clifford gates, local Pauli measurements, and classical communication, if and only if one graph is a so-called vertex minor of the other \cite{dahlberg2018transforming}. Circle graphs, then, are a canonical example of a graph family closed under vertex minors, and the role they play in vertex-minor theory parallels the role that planar graphs play in graph-minor theory \cite{kim2024vertex}.

Circle graphs are also interesting because they are a simple class of vertex-minor-closed graphs with unbounded rank width \cite{kim2024vertex}. This in turn implies that circle graph states have unbounded entanglement width \cite{van2007classical}. Unbounded entanglement width is known to be a necessary condition for MBQC universality of graph states \cite{van_den_nest_universal_2006}, and a family of graph states has unbounded entanglement width if and only if it has all CGSs as vertex minors \cite{geelen_grid_2023}. Thus circle graphs play a deep role in the characterization of MBQC universality. Finding structurally simple universal resource states is valuable, and so this raises the question of whether circle graphs themselves constitute a universal resource family \cite{raussendord2007talk,raussendorf2025PC}.

To address this question, we leverage an unexpected connection to fermionic physics. Stabilizer states, of which graph states form an important subset, exhibit striking structural parallels to another fundamental class of quantum many-body states, known as fermionic Gaussian states (FGSs). Stabilizer states are defined as the joint eigenstates of commuting Pauli operators and underpin quantum error-correcting codes and the Gottesman--Knill theorem \cite{Gottesman1997,AaronsonGottesman2004}. On the other hand, FGSs are ground and thermal states of quadratic fermionic Hamiltonians, are fully characterized by their one-body correlations, and play a central role in free-fermion models of condensed-matter physics and chemistry \cite{Peschel2003}. Both families admit compact algebraic descriptions (binary symplectic vs.\ skew-symmetric matrices) and support efficient updates under their respective linear dynamical groups (Clifford vs.\ matchgate circuits), enabling efficient classical simulability in many contexts \cite{calderbank1997quantum,bravyi2004lagrangian}.

Through the Jordan--Wigner and related mappings, FGSs can be mapped to stabilizer states and vice versa. Such mappings, however, do not preserve important structure such as locality between the fermions and qubits in a meaningful way. Therefore we instead employ a fermion-to-qubit mapping originally introduced to solve the Kitaev honeycomb model \cite{kitaev2006anyons}, which encodes a qubit into four Majorana fermions. We show that this particular mapping elucidates a precise correspondence between certain families of graph states and of FGSs. Using techniques from graph theory and isotropic systems \cite{bouchet_isotropic_1987,bouchet1988graphic,bouchet_connectivity_1989}, we find that this family of graph states is precisely that of circle graphs. In fact, our family of FGSs is expressive enough to represent all quantum states that are local unitary equivalent (LU-equivalent) to CGSs.

This finding allows us to resolve the open question \cite{raussendord2007talk,raussendorf2025PC} of whether circle graphs are universal resource states for MBQC in the negative:

\begin{theorem*}[Informal]\label{thm:main}
    Circle graph states are not universal for MBQC, unless $\BQP = \BPP$.
\end{theorem*}

This result is in line with prior works (see \cref{sec:related_work} for a review) showing that MBQC with certain resource states can be simulated efficiently by classical computers. That is, our proof of non-universality is an efficient classical algorithm for MBQC with CGSs. As such, the notion of universality above should be understood as \emph{efficient} universality~\cite{van2007fundamentals}. Core to our proof is the fact that FGSs admit compact classical descriptions related to free-fermion solvability. 

\subsection{Importance of circle graph states}
\label{sec:importance_circle_graphs}
Local operations and classical communication (LOCC) is a widely used model for understanding entanglement, reflecting spatially separated parties who cannot generate shared entanglement. A more restrictive model is obtained by allowing only local \emph{Clifford} operations and classical communication (LCOCC). For a family of states $\Family$, let $\LOCC(\Family)$ and $\LCOCC(\Family)$ consist of all states reachable from $\Family$ using LOCC or LCOCC operations, respectively. One may observe that the states in $\LOCC(\Family)$ or $\LCOCC(\Family)$ cannot exhibit entanglement that was not already present in $\Family$.

Entanglement width is a structural measure of bottleneck entanglement in a multipartite system, reflecting both the amount and the distribution of correlations throughout the system. More precisely, it is the minimum, over all trees of degree at most $3$ whose leaves correspond to the qubits, of the maximum bipartite entanglement entropy across any edge cut of the tree~\cite{van_den_nest_universal_2006}. For graph states, entanglement width elegantly coincides with the associated graph's \emph{rank width}. Entanglement width is an important measure for MBQC since its divergence is necessary for universality~\cite{van_den_nest_universal_2006}, while families of states with bounded entanglement width can be simulated classically~\cite{van2007classical}. Hence, unless $\BQP = \BPP$, such a family cannot be an efficient universal resource for MBQC. We refer the reader to \cref{sec:EntanglementWidth} for a review of entanglement width and related concepts.

Understanding fundamental families with unbounded entanglement width is therefore a natural strategy for identifying universal resources for MBQC. This is where CGSs step in, being a simple class with unbounded entanglement width. Moreover, if a family of graph states $\Family$ has unbounded entanglement width, then $\LCOCC(\Family)$ must contain every CGS. This is implied by the result of Dahlberg and Wehner~\cite{dahlberg2018transforming} (vertex minors and LCOCC are equivalent for graph states), combined with that of Geelen, Kwon, McCarty, and Wollan~\cite{geelen_grid_2023} (a family of graphs has unbounded rank width if and only if it contains all circle graphs as vertex minors). This parallels one of the main results from Robertson and Seymour's celebrated work on graph minors~\cite{robertson1986graph}:~any family of graphs with unbounded tree width must have every planar graph as a minor. Thus, analogous to the way in which planar graphs signify unbounded tree width, circle graphs signify unbounded rank width.    

Our work settles a natural question:~is the family of circle graph states itself a universal resource?

\subsection{Related work}\label{sec:related_work}

\paragraph{Classical simulation of MBQC.} The universality and the simulability of MBQC with certain resource states, although formally distinct notions, are closely operationally related. Assuming $\BQP \neq \BPP$, a family cannot be efficiently universal~\cite{van2007fundamentals} if it can be classically simulated in polynomial time (note that the converse is not necessarily true). In this paper, we focus on the simulation perspective. Nielsen~\cite{nielsen2006cluster} showed that MBQC on 1D cluster states (i.e., graph states on a regular grid) can be simulated efficiently, which was extended by Yoran and Short~\cite{yoran2006classical} to 2D cluster states of logarithmically bounded width. These results were further subsumed by Markov and Shi~\cite{markov2008simulating} and Jozsa~\cite{jozsa2006simulation}, who generalized to any graph state with logarithmic tree width.

Beyond graph states, Shi, Duan, and Vidal~\cite{shi2006classical} ruled out any state admitting an efficient tree tensor network (TTN) description (i.e., each tensor having polynomially bounded rank). Van den Nest, D{\"u}r, Vidal, and Briegel~\cite{van2007classical} built on this result by efficiently constructing such a TTN for any state whose so-called Schmidt-rank width grows at most logarithmically. Like entanglement width, unbounded Schmidt-rank width is also necessary (but not sufficient) for universality.

The simulability of the above cases can be essentially attributed to the fact that those states are constrained in some important parameter (e.g., tree width, Schmidt-rank width). On the other hand, there are not many known examples where the resource family is defined \emph{a priori} irrespective of such constraints. To our knowledge, the best-known example is the surface code, one of the most popular approaches for achieving fault tolerance within the circuit model~\cite{kitaev2003fault,bravyi1998quantum}. Notably, surface code states exhibit entanglement width growth at least linear in system size~\cite{bravyi2007measurement,van2007classicalspin}, implying that they lie outside the class of states described above. Nonetheless, Bravyi and Raussendorf~\cite{bravyi2007measurement} showed the existence a classical simulation algorithm with surface code resource states on a 2D square lattice, assuming that the measured and unmeasured qubits remain connected at each stage of the MBQC. The techniques of Bravyi, Gosset, and Liu~\cite{bravyi2022simulate} later improved this result, generalizing from square lattices to any planar graph and removing the connectivity condition. Goff and Raussendorf~\cite{goff2012classical} generalized in a different direction, showing that efficient simulation is possible for surface code states embedded in any surface of at most logarithmic genus.

Finally, some works have studied the power of MBQC wherein errors have corrupted an otherwise universal resource state. For example, Browne \emph{et al.}~\cite{browne2008phase} considered qubit losses in the 2D cluster state (the canonical example of a universal resource state). They found that, above some critical threshold for the probability that any given qubit is lost, the MBQC process can be classically simulated. In another model, Atallah \emph{et al.}~\cite{atallah2024efficient} considered perturbing the initial $\ket{+}$ states (used to produce the resource graph state) toward diagonal single-qubit states. Again, they demonstrated a perturbation threshold beyond which the MBQC can be efficiently simulated.

\paragraph{Isotropic systems and graph states.} Isotropic systems were introduced by Bouchet \cite{bouchet_isotropic_1987,bouchet1988graphic,bouchet_connectivity_1989}, and recently they have been appreciated in quantum information \cite{dahlberg_counting_2020} to precisely correspond to stabilizer groups.  Many classic results in the theory of stabilizer codes have simple proofs using the theory of isotropic systems, for instance, that stabilizer states are LC-equivalent to graph states \cite{schlingemann2002stabilizer} can easily be derived from Bouchet's results \cite[Section 3]{bouchet1988graphic}.  Certain isotropic systems called \emph{graphic systems} arise from the notion of a $4$-regular multigraph, and the corresponding graph states are naturally associated to circle graphs \cite{bouchet1988graphic}.  This connection lays the foundation for our work to link FGSs (defined with respect to the $4$-regular multigraph, see \cref{sec:matching_states}) with CGSs.

\paragraph{Connections between stabilizer states and FGSs.} The moral similarities between stabilizer states and fermionic Gaussian states have been heavily studied prior. For example, algorithmic ideas for learning near-stabilizer states have been successfully adapted to learning near-Gaussian states \cite{grewal2024efficient,leone2024learning,hangleiter2024bell,mele2025efficient}. More relevantly, Herasymenko, Anshu, Terhal, and Helsen~\cite{herasymenko2024fermionic} used fermion-to-qubit mappings to shed light on a fermionic version of the no low-energy trivial states (NLTS) problem. Their work proves the existence of local fermionic Hamiltonians without low-energy states prepared by arbitrary Gaussian operations and shallow fermionic circuits, which is in strong analogy to a line of work constructing qubit Hamiltonians without low-energy states of low circuit or low stabilizer complexity \cite{anshu2023nlts,coble2023local,coble2023hamiltonians}. There is also an extensive literature on Majorana stabilizer codes, which give a fermionic view of quantum error correction~\cite{kitaev2001unpaired,bravyi2010majorana,vijay2015majorana,chapman2022free,bettaque2025structure}. 

Beyond making connections between these two classes of states, a recent series of works have begun to explore how the two formalisms enabling their classical simulability can be handled together. It has been known since the work of Jozsa and Miyake~\cite{jozsa2008matchgates} that the union of matchgates (equivalent to fermionic Gaussian unitaries) and Clifford gates generates a universal gate set. However, one can ask what restrictions on their combination yield quantum circuits that remain classically tractable;~this was explored initially by Van den Nest~\cite{van2011simulating}, and later more deeply by Projansky, Necaise, and Whitfield~\cite{projansky2025gaussianity}. Kang \emph{et al.}~\cite{kang2025quon} introduced an alternative framework, dubbed the 2D quon language, under which matchgates and Clifford circuits arise from a single unified structure. As an example application, they construct a new class of Clifford--matchgate hybrid matrix product states using this language.

\section{Technical Overview}

Here we provide a summary of the main ideas of this paper and give a high-level overview of the techniques developed.

\subsection{Kitaev's fermion-to-qubit mapping}

At the core of our results is a mapping between fermions on $2n$ modes and $n$-qubit states. This mapping was first introduced by Kitaev to solve the honeycomb spin model~\cite{kitaev2006anyons}, although we provide a different perspective and application of it. Our construction makes use of several graph-theoretic concepts, which we review in \cref{sec:background}.

First, note that one fermionic mode can be associated with two Majorana operators. We will group $n$ pairs of fermionic modes, and in each pair designate the four Majorana operators by $c_j^1, c_j^2, c_j^3, c_j^4$, $j \in [n]$. Together, these four operators are placed on a single node $j \in [n]$ in a graph. Due to the commutation relations of Majorana operators, we can associate the product of any two Majoranas within a vertex $j$ to a Pauli matrix on a corresponding qubit $j$.

However, the Hilbert space of four Majorana operators has dimension $4$, requiring us to identify this qubit within some two-dimensional subspace. Let $D_j = -c_j^1 c_j^2 c_j^3 c_j^4$, which acts as the parity operator of the two fermions on the $j$th node. The qubit subspace will be given by the $+1$-eigenspace of $D$. Projecting into this subspace, we have for example that $-i c_j^1 c_j^4$ maps to $X$ on qubit $j$. A full table of these products is given in \cref{tab:mapping}, and the formal presentation of this mapping is given in \cref{sec:qubit_mapping}.

\begin{table*}[h!]
\begin{center}
\begin{tabular}{|c|c|c|c|c|}
    \hline
    \rule{0pt}{2.5ex} & $c^1$& $c^2$ & $c^3$ & $c^4$  \\
    \hline
    \rule{0pt}{2.5ex} $c^1$ & $\I$  & $iZ$  & $-iY$  & $iX$ \\
    \hline
    \rule{0pt}{2.5ex} $c^2$ & $-iZ$  & $\I$  & $iX$ & $iY$ \\
    \hline
    \rule{0pt}{2.5ex} $c^3$ & $iY$ & $-iX$  & $\I$  & $iZ$ \\
    \hline
    \rule{0pt}{2.5ex} $c^4$ & $-iX$  & $-iY$  & $-iZ$  & $\I$ \\
    \hline
\end{tabular}
\caption{\label{tab:mapping}The ``multiplication table'' for mapping four Majorana operators to a single qubit. Each entry denotes the product of a row Majorana and a column Majorana, after projecting to the $+1$-eigenspace of the operator $D = -c^1 c^2 c^3 c^4$. For example, $c^1c^4 \mapsto iX$.}
\end{center}
\end{table*}

\subsection{Fermionic states on 4-regular multigraphs}

On each node we have four Majorana operators. It is natural to represent these Majorana operators as half-edges, and we connect them to some other half-edge (Majorana) in the graph. This yields a $4$-regular multigraph, i.e., a graph whose nodes all have degree $4$ and where self-loops and parallel edges (multi-edges) are permitted. As each half-edge is connected to exactly one other half-edge, we can equivalently regard this graph as a perfect matching of $4n$ half-edges. And since each half-edge is a Majorana operator, this perfect matching naturally induces an FGS on $2n$ fermion modes. Gaussian states are related to free-fermion systems, and in particular admit efficient classical descriptions.

We will refer to the Gaussian states $\ket{\Psi}$ obtained by our perfect matchings as \emph{matching states}. We show that such a state is a uniform superposition over a set of $2^{n-1}$ stabilizer states, each of which is encoded into the joint $s_j$-eigenspace ($s_j \in \pmone$) of the operators $D_1, \ldots, D_n$. We can therefore single out one particular stabilizer state by projecting $\ket{\Psi}$ onto the joint $+1$-eigenspace, which we recall is our identified $n$-qubit space.\footnote{The sign of one of the edges need to be flipped so that $\Psi$ has the appropriate projection onto the joint $+1$-eigenspace, see \Cref{def:Psi_hs}.} See \cref{thm:embedded_graph_state} for details.

As an example, consider the GHZ state:
\begin{equation}
    \ket{\mathrm{GHZ}} = \frac{\ket{0}^{\otimes n} + \ket{1}^{\otimes n}}{\sqrt{2}}.
\end{equation}
The stabilizer generators of $\ket{\mathrm{GHZ}}$ can be taken as $X_1 \cdots X_n$ and $Z_j Z_{j+1}$ for $j \in [n-1]$. In \cref{fig:maximally-entangled} we give an explicit construction of a perfect matching on $4n$ half-edges that yields the GHZ state. To see this, take the top (or bottom) loop;~following \cref{tab:mapping}, this gives the Pauli operator $X_1 \cdots X_n$. Meanwhile, a loop connecting a pair of edges between any two neighboring nodes gives $Z_j Z_{j+1}$.

\begin{figure}
    \centering
    \includegraphics[width=6in]{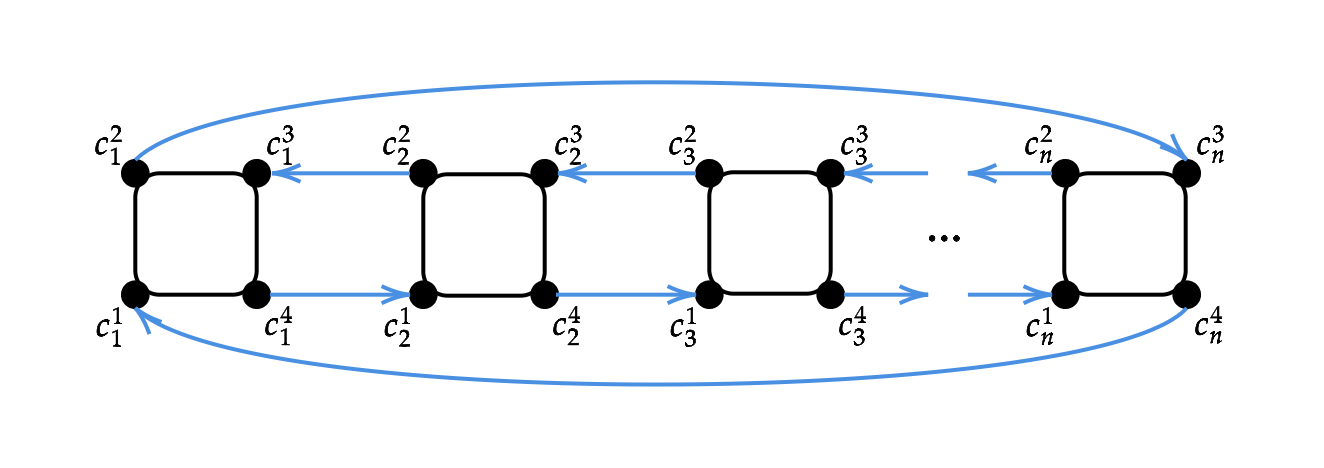}
    \caption{Graphical representation of a GHZ-like state under our mapping. The black squares represent nodes (qubits), and the circular corners correspond to half-edges (Majorana operators). The blue arrows represent edges of a $4$-regular multigraph, with directions given by the orientation of the corresponding Eulerian tour. The top and bottom loops yield the stabilizer $X^{\otimes n}$, while the parallel edges between any two adjacent nodes yield the stabilizer $Z \otimes Z$ (up to sign conventions).}
    \label{fig:maximally-entangled}
\end{figure}

\subsection{Extracting circle graph states from the fermion space}

Next we classify the stabilizer states encoded into the matching state. We show that the $n$-qubit projection is precisely the graph state of a circle graph. This follows from a classic result:~any circle graph can be constructed from the Eulerian tour of some $4$-regular multigraph~\cite{bouchet_connectivity_1989}.\footnote{More precisely, any connected circle graph can be obtained from a $4$-edge connected $4$-regular multigraph.} Note that an Eulerian tour is a (non-simple) cycle through a graph such that every edge is traversed exactly once. This fact enables us to take any $2n$-mode matching state and identify an encoded $n$-qubit CGS. Moreover, we can also efficiently move in the reverse direction:~constructing the perfect matching from the circle graph (see \cref{prop:efficient_chord_diagram}).

We review how this correspondence works at the level of graph theory in \cref{sec:circle_graphs}. To demonstrate the equivalence between the two classes of quantum states, we go a step further and exhibit a concrete representation of the circle graph stabilizers. Recall that the stabilizers of a graph state $\ket{G}$ are generated by the $n$ Pauli operators $S_j = X_j \prod_{k \in \mathcal{N}_G(j)} Z_k$. We can recover these from an Eulerian tour on the $4$-regular multigraph, wherein at each node $j$ we split the tour into two cycles $C_j^1, C_j^2$. To each cycle $C_j^a$ we define the operators $\tilde{S}_j^a$ as the product of Majorana operators (half-edges) traversed along the cycle. Under Kitaev's fermion-to-qubit mapping, we show that both $\tilde{S}_j^1, \tilde{S}_j^2 \mapsto S_j$ within the target $n$-qubit subspace. This is presented rigorously as \cref{thm:cycles_to_stabilizers}.

Note that the results of \cite{bouchet1988graphic} already imply that every circle graph stabilizer group can be obtained in this manner. Our contribution is to 1)~exhibit an explicit matrix representation in terms of the usual Pauli matrices, 2)~show how the cycles on the fermion graph map to these stabilizers, and 3)~articulate how the matching state directly encodes the CGSs.

\subsection{Simulating MBQC}

Having established this fermionic structure, we can prove our main result in \Cref{sec:classical_sim}:~Arbitrary single-qubit measurements on a circle graph state $\ket{G}$ can be efficiently simulated. We accomplish this using the algorithm of \cite{bravyi2022simulate} which only relies on two assumptions:~(1) that we can sample $x \in \{0, 1\}^n$ from the distribution $|\ip{x}{G}|^2$, and (2) that we can compute $|\melem{x}{U_1 \otimes \cdots \otimes U_n}{G}|^2$ for any sequence of single-qubit unitaries $U_j$. In our setting, (1) is standard by stabilizer simulation techniques~\cite{AaronsonGottesman2004}, while (2) is shown by us below. 

We have already outlined above how $\ket{G}$ is encoded in the fermionic space as the matching state $\Psi$. Next we show how to represent any single-qubit state $\ket{\phi}$ in the same fermionic space. Using Kitaev's mapping, we can define a state $\ket{\tilde{\phi}}$ on four Majorana modes whose expectations agree with the Bloch vector of $\ket{\phi}$:~$\ev{(-i c^1 c^4)}{\tilde{\phi}} = \ev{(-i c^2 c^3)}{\tilde{\phi}} = \ev{X}{\phi}$, etc. We prove that this state has all the properties we desire:~it is Gaussian, it is a $+1$-eigenstate of $D$, and its projection onto that eigenspace is precisely $\ket{\phi}$. Extending to product states over $n$ qubits is straightforward:~the state $\op{\tilde{\Phi}}{\tilde{\Phi}} = \prod_{j \in [n]} \op{\tilde{\phi}_j}{\tilde{\phi}_j}$ is also Gaussian, and it projects to $\ket{\Phi} = \bigotimes_{j \in [n]} \ket{\phi_j}$ in the joint $+1$-eigenspace of all $D_1, \ldots, D_n$.

Because we have shown that the projection of $\ket{\Psi}$ onto that $n$-qubit subspace is $\frac{1}{\sqrt{2^{n-1}}} \ket{G}$, we find that $|\ip{\Phi}{G}|^2 = 2^{n-1} |\ip{\tilde{\Phi}}{\Psi}|^2$. But $\ket{\tilde{\Phi}}$ and $\ket{\Psi}$ are both Gaussian, so a standard formula~\cite{bravyi2017complexity} allows us to express this overlap in terms of a Pfaffian:
\begin{equation}\label{eq:gaussian_overlap}
    \abs*{\ip{\tilde{\Phi}}{\Psi}}^2 = \frac{1}{4^n} \pf(\Gamma_\Phi + \Gamma_\Psi),
\end{equation}
where $\Gamma_\Phi, \Gamma_\Psi$ are the $4n \times 4n$ covariance matrices (see \cref{def:cov_mat}) of $\Phi$ and $\Psi$, respectively. The Pfaffian, in turn, can be computed in $\O(n^3)$ time~\cite{wimmer2012algorithm}.

\section{Discussion}

Through our connection between circle graph states and fermionic Gaussian states, we have shown that MBQC with circle graph resource states can be efficiently simulated on classical computers. Beyond the scope of MBQC, we also anticipate broader applications since our correspondence applies more generally between FGSs and many-qubit systems (e.g., beyond stabilizer states). Free-fermion solvability, which was the key to obtaining tractable formulas, has also recently been used in novel ways along with graph-theoretic techniques to find ground states of qubit Hamiltonians \cite{chapman2023unified}. Our work provides another example of this kind of phenomenon. For example, our results imply that the qubit states obtained from FGSs by Kitaev's mapping consist of at least both CGSs and the ground states of the honeycomb model. Identifying other notable classes of states within this framework is a promising direction for future work.

Along these lines, it is important to point out that the prior results on efficiently simulating surface code resource states~\cite{bravyi2007measurement,goff2012classical,bravyi2022simulate} can be fundamentally traced back to free-fermion solvability as well. This follows from how \cite{bravyi2007measurement} reduce the overlap between a surface code state and a product state to the partition function of a 2D Ising model with complex weights. The 2D Ising model, in turn, famously admits an analytical solution initially due to Onsager~\cite{onsager1944crystal}. This classic result has been rederived by others from a variety of different angles~\cite{kaufman1949crystal,kac1952combinatorial,kasteleyn1961statistics,temperley1961dimer,barahona1982computational}, but in particular the seminal paper by Schultz, Mattis, and Lieb~\cite{schultz1964two} showed that the problem can be reduced entirely to that of free fermions. Indeed, this provides an account for why Pfaffians are so prevalent in the surface code. In follow-up work dealing with higher-genus surfaces, this fermionic/matchgate connection is outlined even more explicitly~\cite{goff2012classical}. Given that our results also prove MBQC simulability via free fermions, it would be worthwhile to probe how deeply this connection runs.

From the standpoint of MBQC universality, some important open questions remain. One is the Geelen--McCarty conjecture~\cite{mccarty2021local}, which states that all proper vertex-minor-closed families of graph states can be classically simulated for MBQC. This seems plausible, given that it would be a generalization of our results. Another is the hardness for simulating resource states containing circle-graph-forbidden vertex minors. Our results indicate that such graph states might necessarily contain non-Gaussianity in the fermionic representation, which may be quantified by Gaussian rank or extent~\cite{dias2024classical,reardonSmith2024improved,cudby2023gaussian}. Finally, it would be interesting to characterize under what circumstances the marginal probabilities of CGSs can be efficiently computed. This is motivated by analogy to surface code resource states, wherein despite generic $\sharpP$-hardness~\cite{bravyi2022simulate}, an efficient algorithm exists under certain graph-connectivity constraints~\cite{bravyi2007measurement}.

\paragraph{Note added.} After the preprint of our manuscript was published, the work of Hahn, McCarty, Nautrup, and Claudet~\cite{hahn2026structure} reproduced our main complexity-theoretic result via an alternative proof. Rather than a connection to free fermions, their work is entirely graph-theoretic, reducing to surface code states with quadratic qubit overhead. In contrast, our algorithm only incurs a constant overhead in the number of degrees of freedom. They also sharpen our statement about  unbounded entanglement width of CGSs from $\omega(1)$ to $\Omega(\sqrt{n})$.

\section{Background} \label{sec:background}

We denote the imaginary unit by $i \equiv \sqrt{-1}$. For an integer $k \geq 1$, we define the set $[k] \coloneqq \{1, \ldots, k\}$. We write $\I$ to denote the identity operator on any vector space, which will be clear from context. Denote the unitary and orthogonal groups as $\U(d) \coloneqq \{U \in \C^{d \times d} : UU^\dagger = \I\}$ and $\Orth(d) \coloneqq \{O \in \R^{d \times d} : OO^\T = \I\}$, respectively. Given a subset $S$ of a group, $\langle S \rangle$ is the subgroup generated by all elements of $S$. For a graph $G$, we write $V(G)$ and $E(G)$ for its vertex and edge sets, respectively.

\subsection{Graph states}

Let us review some basic facts about systems of qubits and graph states. First, we introduce the Pauli algebra.

\begin{definition}
    The \emph{Pauli matrices} on a single qubit are
    \begin{equation}
        X = \begin{pmatrix}
            0 & 1\\
            1 & 0
        \end{pmatrix}, \quad
        Y = \begin{pmatrix}
            0 & -i\\
            i & 0
        \end{pmatrix}, \quad
        Z = \begin{pmatrix}
            1 & 0\\
            0 & -1
        \end{pmatrix}.
    \end{equation}
\end{definition}

\begin{fact}\label{fact:pauli_algebra}
    The algebra (over $\C$) of Pauli matrices on a single qubit are generated by the conditions $X^2 = Y^2 = Z^2 = \I$ and $XYZ = i\I$.
\end{fact}

\begin{fact}
    The set of $n$-qubit Pauli operators $\mathcal{P}_n \coloneqq \{\I, X, Y, Z\}^{\otimes n}$ is a complete orthogonal operator basis for $(\C^2)^{\otimes n}$, obeying the trace-orthogonality condition $\tr(P^\dagger Q) = 2^n \delta_{PQ}$ for all $P, Q \in \mathcal{P}_n$.
\end{fact}

The $n$-qubit Pauli operator which acts as $W \in \{X, Y, Z\}$ on qubit $j$ and trivially on all other qubits may be denoted by $W_j$. Next, we define stabilizer groups and states.

\begin{definition}
    Let $\mathcal{S} \subset \{\pm 1\} \cdot \mathcal{P}_n$ be a set of $n$ mutually commuting, algebraically independent Pauli operators. The group $\langle\mathcal{S}\rangle$ is called a \emph{stabilizer group}, and the simultaneous $+1$-eigenstate of all $2^n$ elements of $\langle\mathcal{S}\rangle$ is said to be its \emph{stabilizer state}.
\end{definition}

Every stabilizer state is reachable from $\ket{0}^{\otimes n}$ by applying some Clifford circuit, defined as follows.

\begin{definition}
    Let $\mathcal{P}_n^* \coloneqq \langle \mathcal{P}_n \rangle$ be the $n$-qubit Pauli group and $\mathcal{C}_n^* \coloneqq \{C \in \U(2^n) : C^\dagger \mathcal{P}_n^* C = \mathcal{P}_n^*\}$ its normalizer. We define the \emph{$n$-qubit Clifford group} as $\mathcal{C}_n \coloneqq \mathcal{C}_n^* / \U(1)$.
\end{definition}

\begin{fact}
    The Clifford group can be generated by the gates
    \begin{equation}
        \Hadamard = \frac{1}{\sqrt{2}} \begin{pmatrix}
            1 & 1\\
            1 & -1
        \end{pmatrix}, \quad
        \Phase = \begin{pmatrix}
            1 & 0\\
            0 & i
        \end{pmatrix}, \quad
        \CZ = \begin{pmatrix}
            1 & 0 & 0 & 0\\
            0 & 1 & 0 & 0\\
            0 & 0 & 1 & 0\\
            0 & 0 & 0 & -1
        \end{pmatrix}.
    \end{equation}
\end{fact}

\begin{fact}
    Let $\ket{\psi}$ be an $n$-qubit stabilizer state with stabilizer group $\langle\mathcal{S}\rangle$. There exists some $C \in \mathcal{C}_n$ such that $\ket{\psi} = C\ket{0}^{\otimes n}$, $\mathcal{S} = \{C Z_j C^\dagger : j \in [n]\}$, and
    \begin{equation}
        \op{\psi}{\psi} = \prod_{S \in \mathcal{S}} \l( \frac{\I + S}{2} \r).
    \end{equation}
\end{fact}

Graph states are a special class of stabilizer states. Since stabilizer states are uniquely associated with their stabilizer group (modulo global phases), we can define a graph state in two equivalent ways.

\begin{definition}
    Let $G$ be a simple graph. On each vertex, we place a qubit. Define the \emph{graph state} as
    \begin{equation}
        \ket{G} \coloneqq \prod_{\{v, w\} \in E(G)} \CZ_{vw} \ket{+}^{\otimes |V(G)|},
    \end{equation}
    where $\ket{+} \coloneqq \Hadamard \ket{0} = \frac{\ket{0} + \ket{1}}{\sqrt{2}}$ and $\CZ_{vw}$ is the $\CZ$ gate between qubits $v$ and $w$.
\end{definition}

Clearly, $\ket{G}$ is a stabilizer state because it is generated by a Clifford circuit (Hadamard on each vertex, $\CZ$ on each edge). Its stabilizer group is generated in a graphical manner as well, which gives an alternative definition.

\begin{definition}
    Let $G$ be a simple graph. To each vertex $v \in V(G)$, associate the Pauli operator
    \begin{equation}
        S_v \coloneqq X_v \prod_{w \in \mathcal{N}_G(v)} Z_w
    \end{equation}
    where $\mathcal{N}_G(v) \coloneqq \{w \in V(G) : \{v, w\} \in E(G)\}$ is the neighborhood of $v$ in $G$. Then define the generating set $\mathcal{S}_G \coloneqq \{S_v : v \in V(G)\}$.
\end{definition}

\begin{fact}
    $\langle \mathcal{S}_G \rangle$ is the stabilizer group of the graph state $\ket{G}$.
\end{fact}

Although graph states are a subset of stabilizer states, every stabilizer state is locally Clifford equivalent to some graph state.

\begin{fact}
    For every $n$-qubit stabilizer state $\ket{\psi}$, there exists some graph $G$ on $n$ vertices such that $\ket{\psi} = (C_1 \otimes \cdots \otimes C_n) \ket{G}$ for some local Clifford gates $C_1, \ldots, C_n \in \mathcal{C}_1^*$.
\end{fact}

\subsection{Circle graphs}\label{sec:circle_graphs}

\begin{figure}
    \centering    \includegraphics[width=0.65\linewidth]{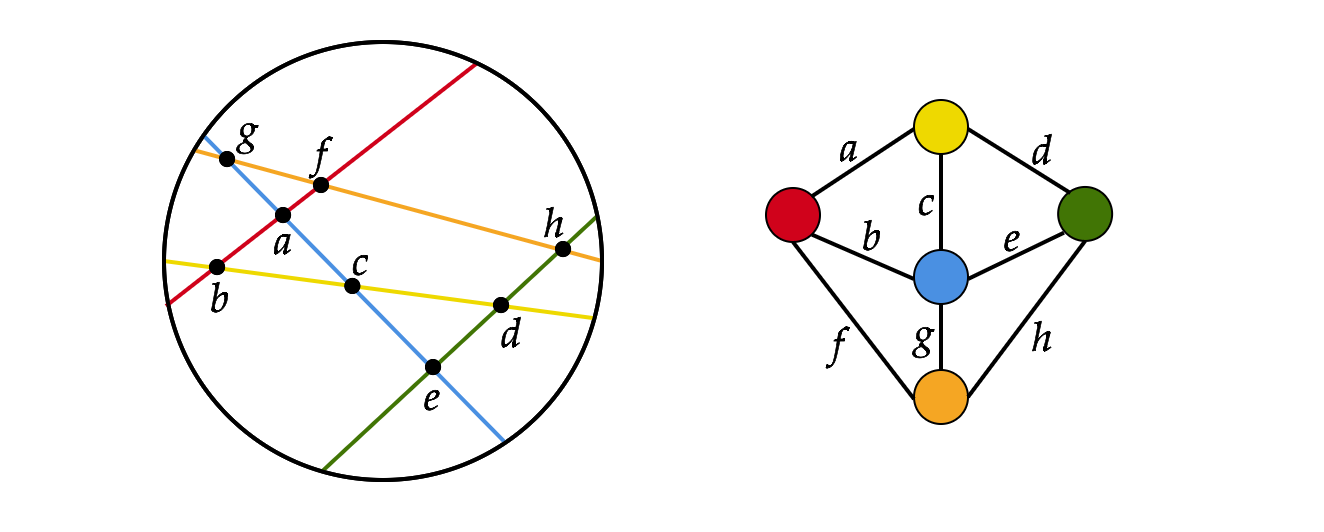}
    \caption{A chord diagram (left) and its corresponding circle graph (right). Chords and their intersections in the chord diagram correspond to vertices and edges, respectively, in the circle graph. Note that we label chords by color and intersections by letters.}
    \label{fig:Circle_graph}
\end{figure}

In this paper, we are centrally concerned with graph states over a special class of graphs, called circle graphs. We review the properties of this class of graphs and its connection to $4$-regular multigraphs, which have been studied extensively in~\cite{bouchet1988graphic}.

We begin by defining a circle graph as the intersection graph of a chord diagram.

\begin{definition}
    Consider a set of chords on a circle. A \emph{circle graph} $G$ is a simple graph where $V(G)$ is the set of chords, and two vertices have an edge in $E(G)$ if and only if their corresponding chords intersect.
\end{definition}

See \cref{fig:Circle_graph} for an example of a chord diagram and its corresponding circle graph. A \emph{circle graph state} (CGS) is then simply a graph state $\ket{G}$ where $G$ is a circle graph.

We can equivalently view a circle graph as the alternance graph of some double-occurrence word. Let us define these terms below.

\begin{definition}
    Let $V$ be a set. A \emph{word} $A$ on $V$ is an ordered sequence of elements from $V$ (the ``letters''), possibly with repetition.
\end{definition}

When the set of letters is clear from context, we may omit specifying it. We will view the chords of the circle as letters, and construct a word given by the order of chords encountered as we traverse the circumference of the circle. Thus it is useful to introduce a notion of cyclic equivalence for words.

\begin{definition}
    Let $A$ be a word and $\widetilde{A}$ be the word whose letters are the letters of $A$ in reverse order. We say that two words $A, B$ are equivalent ($A \sim B$) if and only if $B$ is some cyclic permutation of $A$ or $\widetilde{A}$. We denote the equivalence class of $A$ with respect to this relation by
    \begin{equation}
        [A] \coloneqq \{B : B \sim A\}.
    \end{equation}
\end{definition}

Next, we define the notion of a double-occurrence word.

\begin{definition}
    A \emph{double-occurrence word} is a word $A$ with letters from some set $V$ such that each letter $v \in V$ occurs exactly twice in $A$.
\end{definition}

For example, the word ``$aabcbcdeed$'' is a double-occurrence word. Because each letter in a double-occurrence word appears exactly twice, we can talk about whether or not those two occurrences interleave, or alternate, each other.

\begin{definition}
    Let $A$ be a word on $V$. An \emph{alternance} of $[A]$ is a pair of distinct letters $(v, w) \in V \times V$ such that a double-occurrence word of the form $\ldots v \ldots w \ldots v \ldots w \ldots$ is in $[A]$.
\end{definition}

Note that since $A \sim \widetilde{A}$, if $(v, w)$ is an alternance of $A$, so too is $(w, v)$. The alternance graph of a double-occurrence word is then the graph whose edges correspond to letters that alternate each other in the word.

\begin{definition}
    Let $A$ be a double-occurrence word on $V$. The \emph{alternance graph} $G$ of $A$ is a simple graph with vertices $V(G) = V$ and edges $E(G) = \{\{v, w\} \subseteq V(G) : (v, w) \text{ is an alternance of } [A]\}$.
\end{definition}

Constructing a double-occurrence word from a chord diagram as described above, one can see that two letters in the word are an alternance if and only if their corresponding chords intersect. Thus circle graphs are precisely alternance graphs.

\begin{fact}
    A graph $G$ is a circle graph if and only if it is the alternance graph of some double-occurrence word.
\end{fact}

One can check this from \cref{fig:Circle_graph}, where the induced double-occurrence word is $RBOYRGBYOG$ ($R = \text{red}$, $B = \text{blue}$, etc.). For example, $R$ and $B$ alternate in the word, and they cross in the chord diagram at point $a$. Meanwhile, $R$ and $G$ do not alternate in the word, and we see that they also do not intersect in the chord diagram.

\subsection{Eulerian tours of 4-regular multigraphs}

Every circle graph can be obtained by traversing an Eulerian tour of a $4$-regular multigraph. This result will be essential to connect circle graphs with fermionic states. Here we review the concepts necessary to understand this graph-theoretic result.

\begin{definition}
    A half-edge in a graph is an incidence of an edge on a vertex.
\end{definition}

Thus each edge in a graph corresponds to two half-edges:~one for each vertex it connects. This is useful for graphs where edges may be parallel or loop onto the same vertex.

\begin{definition}
    A \emph{multigraph} is a graph $F$ that contains self-loops and multi-edges. To distinguish its edges, each $e \in E(F)$ is described by two half-edges, $e = \{h_v^\alpha, h_w^\beta\}$ where $v, w \in V(F)$ and $\alpha \in [\deg(v)], \beta \in [\deg(w)]$ are labels to distinguish different half-edges on the same vertex.
\end{definition}

\begin{definition}
    A graph $F$ is \emph{$d$-regular} if every $v \in V(F)$ has $\deg(v) = d$.
\end{definition}

Hence a $4$-regular multigraph is a graph where each vertex contains exactly $4$ half-edges. Next, we define the Eulerian tour.

\begin{definition}\label{def:tour}
    A \emph{tour} $T$ on a graph $F$ is an alternating sequence of vertices and edges,
    \begin{equation}
        T = v_1 e_1 v_2 e_2 \ldots v_\ell e_\ell v_{\ell+1},
    \end{equation}
    such that each edge $e_j = \{h_{v_j}^\alpha, h_{v_{j+1}}^\beta\}$ is not repeated and $v_{\ell+1} = v_1$.  An \emph{Eulerian tour} is a tour that covers every edge of $F$.
\end{definition}

\begin{remark}
    When a tour is not empty, we may refer to it as a \emph{cycle}, in the sense that it induces an even-degree subgraph.
\end{remark}

$4$-regular multigraphs are special with regard to Eulerian tours. Without loss of generality we assume a connected multigraph, as otherwise we can work with each connected component at a time.

\begin{fact}
    Let $F$ be a connected $4$-regular multigraph. Then $F$ is guaranteed to have at least one Eulerian tour, and in every Eulerian tour of $F$, each vertex is traversed exactly twice (except for the first vertex $v_1$, which is visited thrice by the convention of \cref{def:tour}).
\end{fact}

Such a tour therefore induces a double-occurrence word whose letters correspond to each vertex as it is traversed in the tour, disregarding the final visit to the starting vertex. Formally, this word is the vertex sequence of the tour.

\begin{definition}
    Let $T = v_1 e_1 v_2 e_2 \ldots v_\ell e_\ell v_{\ell+1}$ be a tour on $F$. The \emph{vertex sequence} $\VS(T)$ of $T$ is the following word on $V(F)$:
    \begin{equation}
        \VS(T) \coloneqq v_1 v_2 \ldots v_\ell.
    \end{equation}
\end{definition}

\begin{fact}
    Let $F$ be a $4$-regular multigraph and $T$ an Eulerian tour. The vertex sequence $\VS(T)$ is a double-occurrence word on $V(F)$.
\end{fact}

Thus since $T$ induces a double-occurrence word on the vertices of $F$, there exists a circle graph $G$, with $V(G) = V(F)$, which is the alternance graph of $\VS(T)$. There is a very elegant way to see the connectivity of a circle graph from its presentation as a $4$-regular multigraph. This involves the notion of a split of an Eulerian tour.

\begin{definition}\label{def:split}
    Let $T = v_1 e_1 v_2 e_2 \ldots v_\ell e_\ell v_{\ell+1}$ be an Eulerian tour. The \emph{split} of $T$ at vertex $v = v_j = v_k$ is the graph operation that transforms $T$ into the two cycles $C_v^1 = \ldots v_{j-1} e_{j-1} v_j e_k v_{k+1} e_{k+1} v_{k+2} \ldots$ and $C_v^2 = \ldots v_{k-1} e_{k-1} v_k e_j v_{j+1} e_{j+1} v_{j+2} \ldots$\,.
\end{definition}

\begin{fact}
    Let $F$ be a $4$-regular multigraph, $T$ an Eulerian tour, and $G$ the alternance graph of $\VS(T)$. Consider the split of $T$ at $v \in V(F) = V(G)$. Then for all $w \in V(G)$, $\{v, w\} \in E(G)$ if and only if $w$ appears exactly once in each cycle $C_v^1, C_v^2$.
\end{fact}

Given any $4$-regular multigraph, it is straightforward to draw some Eulerian tour and hence construct an associated circle graph. It turns out that the reverse direction---constructing a $4$-regular multigraph from a given circle graph---can be accomplished efficiently as well.

\begin{proposition}\label{prop:efficient_chord_diagram}
    Let $G$ be a simple graph. There exists an efficient algorithm to decide if $G$ is a circle graph, and if so, that also constructs a $4$-regular multigraph $F$ and an Eulerian tour $T$ on $F$ such that $G$ is the alternance graph of $\VS(T)$.
\end{proposition}

The currently best-known algorithm is due to \cite{gioan2014practical}, which runs in time essentially linear in $|V(G)| + |E(G)|$. In particular, it builds a representation of a chord diagram to check whether $G$ is circle or not. Assuming the chord diagram is valid, we can use it to easily construct the associated tour and multigraph. Earlier works that also give polynomial-time algorithms for circle graph recognition, albeit with higher scalings, include~\cite{naji1985reconnaissance,bouchet1987reducing,gabor1989recognizing,spinrad1994recognition}.

\subsection{Fermionic Gaussian states}\label{Sec: GaussianFermionicStates}

Here we review the necessary background for fermionic systems, and free-fermion (Gaussian) states in particular. We refer the reader to \cite{bravyi2004lagrangian} for a more thorough exposition.

Consider a system of $N$ fermionic modes, which lives in a Hilbert space (called Fock space) of dimension $2^N$. The algebra of operators over this space can be generated by $2N$ Majorana operators.

\begin{definition}
    Define operators $c_1, \ldots, c_{2N}$ on $\C^{2^N}$ which obey the canonical anticommutation relation:
    \begin{equation}
        c_j c_k + c_k c_j = 2 \delta_{jk} \I.
    \end{equation}
    Furthermore for any $Q \subseteq [2N]$, we define the monomials
    \begin{equation}
        c_Q \coloneqq \prod_{j \in Q} c_j,
    \end{equation}
    where we take the convention that indices increase from left to right in the product. We call all such operators \emph{Majorana operators}, and $Q$ is said to be the \emph{support} of $c_Q$.
\end{definition}

More generally, we will say that the support of an operator $A = \sum_{Q \subseteq [2N]} A_Q c_Q$ is $\bigcup_{Q : A_Q \neq 0} Q$. We will later see that this notion is basis-dependent. Note that any operator can be decomposed in this way, with $A_Q = 2^{-N} \tr(c_Q^\dagger A)$. This is a consequence of the following fact.

\begin{fact}
    The Majorana operators form a complete orthogonal operator basis,
    \begin{equation}
        \tr(c_Q^\dagger c_R) = 2^N \delta_{QR}.
    \end{equation}
    In fact, $\{c_Q : Q \subseteq [2N]\}$ obeys precisely the $N$-qubit Pauli algebra.
\end{fact}

Being isomorphic to Pauli operators, Majorana operators also either commute or anticommute. Specifically, we have the following relation.

\begin{fact}\label{fact:majorana_commute}
    For $Q, R \subseteq [2N]$,
    \begin{equation}
        c_Q c_R = (-1)^{|Q| \cdot |R| + |Q \cap R|} c_R c_Q.
    \end{equation}
\end{fact}

In a system of fermions, parity superselection rules dictate that any physical observable or density operator is supported on only even-degree Majorana monomials. Thus we are primarily concerned with both $|Q|$ and $|R|$ being even, in which case $c_Q$ and $c_R$ commute if and only if their supports intersect an even number of times. We also have that the support of $c_Q c_R$ is equal to the symmetric difference $Q \triangle R$.

Majorana operators can be rotated into each other by a Bogoliubov transformation called a fermionic Gaussian unitary.

\begin{definition}\label{def:FGU}
    Let $O \in \Orth(2N)$. A unitary operator $U_O$ obeying
    \begin{equation}
        U_O^\dagger c_j U_O = \sum_{k \in [2N]} O_{jk} c_k
    \end{equation}
    for all $j \in [2N]$ is called a \emph{fermionic Gaussian unitary}.
\end{definition}

Next we define fermionic Gaussian states. A natural physics-based definition is as the thermal and ground states of free-fermion Hamiltonians. An equivalent definition prescribes the form of the density operator explicitly.

\begin{definition}\label{def:FGS}
    Let $\rho$ be a density operator. We say $\rho$ is a \emph{fermionic Gaussian state} (FGS) if there exists $O \in \Orth(2N)$ and a sequence of real numbers $\lambda_1, \ldots, \lambda_N \in [-1, 1]$ such that
    \begin{equation}\label{eq:gaussian_state_defn}
        \rho = U_O \prod_{j \in [N]} \l( \frac{\I - i \lambda_j c_{2j-1} c_{2j}}{2} \r) U_O^\dagger.
    \end{equation}
\end{definition}

One immediately sees that an FGS is pure if and only if all $|\lambda_j| = 1$. A convenient formalism when working with Gaussian states is the notion of a covariance matrix, which is of polynomial size.

\begin{definition}\label{def:cov_mat}
    The \emph{covariance matrix} of an $N$-mode state $\rho$ is the $2N \times 2N$ real, skew-symmetric matrix $\Gamma_{\rho}$ whose entries are defined as
    \begin{equation}
        [\Gamma_\rho]_{jk} \coloneqq -\frac{i}{2} \tr([c_j, c_k] \rho).
    \end{equation}
\end{definition}

If we have the covariance matrix of some state $\rho$, then following \cref{def:FGU} the covariance matrix of the rotated state $U_O \rho U_O^\dagger$ for any orthogonal $O$ can be found as:
\begin{equation}
    \Gamma_{U_O \rho U_O^\dagger} = O \Gamma_\rho O^\T.
\end{equation}
We can completely characterize covariance matrices as follows.

\begin{fact}\label{fact:cov_mat_decomp}
    A matrix $\Gamma \in \R^{2N \times 2N}$ is the covariance matrix of some state $\rho$ if and only if $\Gamma = -\Gamma^\T$ and $\Gamma \Gamma^\T \preceq \I$. Any such matrix $\Gamma = \Gamma_\rho$ can be brought to canonical form
    \begin{equation}
        \Gamma_\rho = O \bigoplus_{j \in [N]} \begin{pmatrix}
            0 & \lambda_j\\
            -\lambda_j & 0
        \end{pmatrix} O^\T, \quad \text{where } \lambda_j \in [-1, 1] \text{ and } O \in \Orth(2N),
    \end{equation}
    and furthermore one can choose $O$ and $\lambda_j$'s to be the same as in \cref{eq:gaussian_state_defn} when $\rho$ is Gaussian. In particular, $\rho$ is pure and Gaussian if and only if $\Gamma_\rho$ is skew-symmetric and orthogonal ($\Gamma \Gamma^\T = \I$).
\end{fact}

Thus the coefficients $\lambda_j$ in a Gaussian state are the \emph{Williamson eigenvalues} of its covariance matrix. The final point of \cref{fact:cov_mat_decomp} gives a necessary and sufficient condition on the covariance matrix of a pure Gaussian state, but not for mixed ones. Instead, a more general condition for $\rho$ to be Gaussian involves Wick's theorem.

\begin{fact}\label{fact:wick}
    A state $\rho$ is Gaussian if and only if its covariance matrix obeys Wick's theorem:
    \begin{equation}
        \tr(c_Q \rho) = i^{|Q|/2} \pf\l(\Gamma_\rho[Q]\r) \quad \forall Q \subseteq [2N],
    \end{equation}
    where $\Gamma_\rho[Q]$ is the principal submatrix of $\Gamma_\rho$ corresponding to rows and columns $j \in Q$.
\end{fact}

Above, $\pf(A)$ is the Pfaffian of a skew-symmetric matrix $A$, which can be computed in $\O(N^3)$ time~\cite{wimmer2012algorithm}. Pfaffians are closely related to the determinant via the identity
\begin{equation}
    \pf(A)^2 = \det(A).
\end{equation}

\section{Fermionic Matching States}\label{sec:matching_states}

\subsection{Identifying qubits within subspaces of fermions}\label{sec:qubit_mapping}

We now aim to give an explicit representation of the encoded qubits. The fermion-to-qubit mapping that achieves this was first introduced by Kitaev~\cite{kitaev2006anyons} to solve his honeycomb spin model.

\begin{definition}[Kitaev mapping]
    For each $v \in [n]$, define the (suggestively named) operators
    \begin{equation}\label{eq:qubit_mapping}
        \tilde{X}_v \coloneqq -ic_v^1 c_v^4, \quad \tilde{Y}_v \coloneqq -ic_v^2 c_v^4, \quad \tilde{Z}_v \coloneqq -ic_v^3 c_v^4.
    \end{equation}
    Also define the \emph{local gauge operators}
    \begin{equation}
        D_v \coloneqq -c_v^1 c_v^2 c_v^3 c_v^4.
    \end{equation}
\end{definition}

Being isomorphic to Pauli operators, each $D_v$ splits the fermionic Hilbert space in half according to their $\pm 1$-eigenspaces. We can identify a qubit within each of these subspaces. This is shown by the following statement, temporarily dropping subscripts for ease of presentation. We refer to \cref{tab:mapping} for a visual description of the mapping within the $+1$-eigenspace.

\begin{claim}\label{claim:pauli_subspace_mapping}
    Let $\subspace{s} \coloneqq \{ \ket{\psi} : D\ket{\psi} = s\ket{\psi} \}$ be the $s$-eigenspace of $D$, where $s \in \pmone$. Within $\subspace{s}$, we have the equivalences 
    \begin{equation}\label{eq:pauli_equiv}
        \tilde{X} \cong sX, \quad \tilde{Y} \cong sY, \quad \tilde{Z} \cong sZ,
    \end{equation}
    as well as
    \begin{equation}\label{eq:qubit_mapping_aux}
        -ic^2 c^3 \cong sX, \quad -ic^1 c^3 \cong -sY, \quad -ic^1 c^2 \cong sZ.
    \end{equation}
\end{claim}

\begin{remark}[Equivalence within subspaces]
    The notation $A \cong B$ for operators $A, B$ indicates that there is a projector $\Pi$ such that $\Pi A \Pi = B \oplus \bm{0}$, where $\bm{0}$ is the zero operator on the image of $\I - \Pi$. We use this equivalence analogously for state vectors (modulo some global phase). The subspace within which the equivalence holds will be clear from context.
\end{remark}

\begin{proof}[Proof (of \cref{claim:pauli_subspace_mapping})]
    We aim to show that the operators obey the Pauli algebra (\cref{fact:pauli_algebra}) within the subspace. By inspection, $\tilde{X}$, $\tilde{Y}$, and $\tilde{Z}$ square to the identity. It remains to show the condition $XYZ = i\I$. Observe that
    \begin{equation}
        \tilde{X} \tilde{Y} \tilde{Z} = -i c^1 c^2 c^3 c^4.
    \end{equation}
    But within $\subspace{s}$, we have $D = -c^1 c^2 c^3 c^4 \cong s\I$. Thus
    \begin{equation}
        \tilde{X} \tilde{Y} \tilde{Z} \cong is\I,
    \end{equation}
    which agrees with the identity $XYZ = i\I \implies (sX)(sY)(sZ) = is\I$.

    To show \cref{eq:qubit_mapping_aux}, multiply by $sD \cong \I$:
    \begin{equation}
        -i c^2 c^3 \cong -i c^2 c^3 \cdot sD = is c^2 c^3 c^1 c^2 c^3 c^4 = -is c^1 c^4 \cong sX.
    \end{equation}
    The other two equivalences can be shown by a similar calculation.
\end{proof}

\subsection{Construction of the state from a multigraph}\label{sec:construction_of_matching_state}

Let $F$ be a $4$-regular multigraph on $n$ vertices. Throughout, we assume that $F$ is loop-free and connected. The former assumption is valid without loss of generality, because any self-loop on some vertex $v$ in $F$ immediately implies that $v$ is disconnected from all other vertices in the corresponding circle graph $G$. The latter assumption simplifies our presentation, as we may consider each connected component one at a time. Note that two disjoint components in $F$ give rise to at least two disjoint components in $G$, implying that the graph state $\ket{G}$ is separable across that bipartition.

A $4$-regular multigraph naturally induces a perfect matching on $4n$ half-edges. We can further assign each edge with a direction, given by the orientation of some Eulerian tour $T$ on $F$. (Recall that the circle graph is obtained from $F$ by drawing a particular Eulerian tour.) This perfect matching structure then naturally induces a class of FGSs, which we describe here.

Half-edges in $F$ will be labeled as $h_v^\alpha$ where $v$ indicates one of the incident vertices $\alpha\in [4]$ is some labeling of the four half-edges on $v$.  Given a tour $T$ of the multigraph $F$ we define the half-edge sequence, $\HS(T)$, as the sequence of half-edges traced out by the tour using the following labeling convention for the $\alpha$ parameters.

\begin{definition}[Half-edge sequence of a tour]\label{def:half_edge_labels}
Fix some vertex $v_1$ as the starting point of the Eulerian tour $T$ and let it have vertex label $1$.  Set the first half-edge traversed in the tour to be $h_{1}^4$ and the last edge traversed in the tour to be $h_{1}^3$.  Let the half-edge corresponding to the first return of the tour to $v_1$ have label $h_{1}^1$ and the following edge be labeled $h_{1}^2$.  For the remaining vertices we follow the tour and use the lowest currently unused $v, \alpha$ labels for each half-edge we meet.  For example, a tour which has vertex sequence $\VS(T) = vwvw$ will have
    \begin{equation}
        \HS(T) = h_v^4 h_w^1 h_w^2 h_v^1 h_v^2 h_w^3 h_w^4 h_v^3.
    \end{equation}
\end{definition}

Now let us define a fermionic density operator by identifying each half-edge with a Majorana operator, and each commuting projector with an edge in the tour.

\begin{definition}[Matching states]\label{def:Psi_hs}
    Let $T$ be an Eulerian tour of a $4$-regular multigraph $F$ on $n$ vertices with half-edge sequence $\HS(T)=h_{1}^4 h_{2}^1 h_{2}^2 \ldots \,$. Let $E = \{(h_v^\alpha, h_w^\beta)\}_{v,w;\alpha,\beta}$ be the directed edges traversed in the tour where $h_v^\alpha$ comes before $h_w^\beta$.  For each half-edge $h_v^\alpha$ define a Majorana operator $c_v^\alpha$.  Define the \emph{matching state} $\Psi$ stabilized by $i c_1^4 c_2^1$ and by $-i c_i^j c_k^l$ for all $(h_i^j,h_k^l)\in E\setminus (h_1^4, h_2^1)$:
    \begin{equation}\label{eq:matching_state_defn}
        \Psi \coloneqq \l( \frac{\I + i c_1^4 c_2^1}{2} \r) \prod_{(h_v^\alpha, h_w^\beta) \in E\setminus (h_1^4, h_2^1)} \l( \frac{\I - i c_v^\alpha c_w^\beta}{2} \r).
    \end{equation}
\end{definition}

Clearly, this state is a stabilizer state because of the perfect matching structure:~all edges hit an even number of Majoranas and have disjoint supports, so the operators all commute and are independent group generators. Because $\Psi$ is a pure state, we will occassionally write $\ket{\Psi}$ such that $\Psi = \op{\Psi}{\Psi}$. 

\begin{remark}[Sign convention]
    One can imagine almost all the stabilizers of $\Psi$ as generated by $-i c_i^j c_k^l$ where the tour $T$ traverses the edge $h_i^j \rightarrow h_k^l$.  Only the sign of the first operator is flipped relative to the direction of the tour.  The reason for this is to ensure that the projection of $\Psi$ onto the subspace of $D_v=+1$ for all $v$ yields the desired graph state.  We could have chosen to flip any odd number of edges and achieved the same result.  Essentially for defining the stabilizer state, we are flipping a single edge in the directed perfect matching corresponding to the Eulerian tour.  
\end{remark}

\begin{definition}[Edge operators]\label{def:edge_ops}
    Given a graph $F$ and a tour $T$, the stabilizers of $\Psi$ described in \Cref{def:Psi_hs} will be referred to as the \emph{edge operators} $\edgeops$.
\end{definition}

Now we state some simple but important facts about $\Psi$.

\begin{claim}
    The matching state $\Psi$ is a pure Gaussian state, with covariance matrix equal to the skew-adjacency matrix of the directed perfect matching:
    \begin{equation}
       [\Gamma_{\Psi}]_{h_v^\alpha, h_w^\beta} = \begin{cases}
           1 & (h_v^\alpha, h_w^\beta) \in E \setminus (h_1^4, h_2^1),\\
           -1 & (h_v^\alpha, h_w^\beta) = (h_1^4, h_2^1),\\
           -[\Gamma_{\Psi}]_{h_w^\beta, h_v^\alpha} & (h_w^\beta, h_v^\alpha) \in E,\\
           0 & \text{else}.
       \end{cases}
    \end{equation}
\end{claim}

\begin{proof}
    Recall \cref{def:FGS} for the definition of a Gaussian state. \cref{eq:matching_state_defn} can be brought to this form by a permutation matrix $O \in \Orth(4n)$ which bijectively maps each directed edge $(h_v^\alpha, h_w^\beta)$ to the canonical pair $(2j-1, 2j)$.  The resulting density matrix has the form of \Cref{eq:gaussian_state_defn} with $\lambda_j \in \pmone$ for all $j$. The covariance matrix follows by \cref{def:cov_mat}. 
\end{proof}

\begin{claim}\label{claim:prod_of_gauges}
    The matching state $\Psi$ is in the $+1$-eigenspace of the global parity operator $\prod_{j=1}^n D_j$.
\end{claim}
\begin{proof}
    The state $\Psi$ is stabilized by the product of all the edge operators in $\edgeops$. Let $A$ be the operator defined by multiplying all edge operators in the order that they appear in the tour $T$.  Since a $4$-regular graph on $n$ vertices must have $2n$ edges, we can factor out the phases signs to obtain an overall factor of $i(-i)^{2n-1}=(-1)^{n-1}$.  By our edge-numbering convention for all vertices $v\neq 1$, $c_v^1$ and $c_v^2$ appear only once in $A$ and are adjacent.  Hence we may commute the pair $c_v^1 c_v^2$ through any other Majorana operator in $A$. The same holds for $c_v^3$ and $c_v^4$. In particular, if the vertices are numbered $1, \ldots, n$ we may commute pairs of operators through to obtain:
    \begin{equation}
        A=(-1)^{n-1} c_1^4 c_1^1 c_1^2 c_1^3 \prod_{j=2}^n c_j^1 c_j^2 c_j^3 c_j^4=(-1)^n \prod_{j=1}^n c_j^1 c_j^2 c_j^3 c_j^4=\prod_j D_j.
    \end{equation}
    Equivalently, this also implies that $\pf(\Gamma_\Psi) = 1$ by Wick's theorem (\cref{fact:wick}).
\end{proof}

Finally, we remark that our description so far has been in terms of the multigraph $F$ and tour $T$. However because we can find these efficiently given a circle graph $G$ via \cref{prop:efficient_chord_diagram}, we can efficiently construct $\Psi$ starting from $G$ as well.

\subsection{Graph state stabilizers from cycles}\label{sec:stabilizers_from_cycles}

Now we show how to derive graph state stabilizers from $\Psi$. Cycles on the graph $F$ derived from the tour $T$ will be used to construct cycle operators.  Using the encoding of \Cref{sec:qubit_mapping} these cycle operators will correspond to stabilizers of the graph state.  To begin, we recall the notion of splitting a tour, which was defined in \cref{def:split}. Let us restate it here, adapting the notation to our labeling convention for half-edges.

\begin{definition}[Split of a tour;~fixed labeling convention]
    Let $F$ be a connected $4$-regular multigraph and $T$ be an Eulerian tour on $F$. The \emph{split} of $T$ at a vertex $v \in V(F)$ is the graph operation that replaces the transitions $h_v^1 h_v^2$ and $h_v^3 h_v^4$ in $T$ with $h_v^1 h_v^4$ and $h_v^3 h_v^2$.
\end{definition}

Recall that this operation splits the tour into two cycles, say $C_v^1$ and $C_v^2$, which preserve the orientation of $T$. The following properties of these cycles will be crucial.

\begin{lemma}\label{lem:cycle_properties}
     Let $F$ be a connected $4$-regular multigraph and $T$ be an Eulerian tour on $F$. Let $G$ be the circle graph induced by $F$ and $T$. The split of $T$ at vertex $v \in V(F)$ yields two cycles $C_v^1$ and $C_v^2$ with the following properties:
     \begin{enumerate}
        \item If $w \in V(F)$ is not adjacent to $v$ within $G$, then $w$ appears twice in $\VS(C_v^1)$ and zero times in $\VS(C_v^2)$ (or vice versa).
        \item If $w \in V(F)$ is adjacent to $v$ within $G$, or $w = v$, then it appears exactly once in each cycle.
        \item For each $w \neq v$ which appears in a cycle, the transitions $h_w^1 h_w^2$ and/or $h_w^3 h_w^4$ are traversed. Otherwise on vertex $w = v$, either the transition $h_v^1 h_v^4$ or $h_v^3 h_v^2$ are traversed.
    \end{enumerate}
\end{lemma}

\begin{proof}
    Since $F$ is a $4$-regular graph and since the tour $T$ exits a vertex immediately after it enters (since $F$ has no self loops) all vertex transitions in $T$ must be the form $h_v^1 \rightarrow h_v^2$ or $h_v^3 \rightarrow h_v^4$.  We also note that since the graph is $4$-regular that the transtions $h_v^1 \rightarrow h_v^2$ and $h_v^3 \rightarrow h_v^4$ each appear once in $\HS(T)$.  By convention, for an arbitrary vertex $v$ the half-edge $h_v^3$ comes after $h_v^2$ in $\HS(T)$.  Hence taking the transition from $h_v^3$ to $h_v^2$ provides a cycle inside the tour:
\begin{equation}
\HS(T)=\ldots\,
h_{v}^{2}\tikzmark{hv2l} \,\,\,\tikzmark{hv2r}
\,\ldots\,
\tikzmark{hv3}h_{v}^{3}
\,h_{v}^{4}\,\ldots
\end{equation}
\begin{tikzpicture}[remember picture,overlay]
\draw[->,thick,bend right=30]
([yshift=2ex]pic cs:hv3) to ([yshift=2ex]pic cs:hv2l);
\draw[->,thick]
([yshift=-1.5ex,xshift=0.5ex]pic cs:hv2l)
  -- 
([yshift=-1.5ex,xshift=0.5ex]pic cs:hv2r);
\end{tikzpicture}
Similarly, for all vertices except the first $h_v^4$ will always come after $h_v^2$ and $h_v^3$ so the transition from $h_v^1$ to $h_v^4$  will always create a cycle which ``wraps around'' the tour:

\begin{equation}
\HS(T)=\tikzmark{ell1l}\ldots\tikzmark{ell1r}\;
\tikzmark{hv1l}h_{v}^{1}\tikzmark{hv1r}\;
h_{v}^{2}\;\ldots\;h_{v}^{3}\;
\tikzmark{hv4l}h_{v}^{4}\tikzmark{hv4r}\;
\tikzmark{ell2l}\ldots\tikzmark{ell2r}
\end{equation}

\begin{tikzpicture}[remember picture,overlay]
\draw[->,thick,bend left=30]
([yshift=2.5ex]pic cs:hv1r) to ([yshift=2.5ex]pic cs:hv4l);

\draw[->,thick]
([yshift=-1.5ex]pic cs:ell1l) -- ([yshift=-1.5ex]pic cs:ell1r);

\draw[->,thick]
([yshift=-1.5ex]pic cs:ell2l) -- ([yshift=-1.5ex]pic cs:ell2r);
\end{tikzpicture}
For the first vertex $h_1^4$ comes before $h_1^1$, so the cycle which takes transition $h_1^1 \rightarrow h_1^4$ does not wrap around the tour but still splits $T$ into edge disjoint cycles:
\begin{equation}
\HS(T)= h_{1}^{4} \tikzmark{m4l} \,\,\,\tikzmark{m4r}
\;\ldots\;
\tikzmark{m1}h_{1}^{1}\;h_{1}^{2}\;\ldots
\end{equation}
\begin{tikzpicture}[remember picture,overlay]
\draw[->,thick,bend right=30]
([yshift=2ex]pic cs:m1) to ([yshift=2ex]pic cs:m4l);

\draw[->,thick]
([yshift=-1.5ex]pic cs:m4l) -- ([yshift=-1.5ex]pic cs:m4r);
\end{tikzpicture}
Let us define $C_v^1$ to always be the cycle with transition $h_v^1 \rightarrow h_v^4$ and $C_v^2$ to always be the cycle with transition $h_v^3 \rightarrow h_v^2$.  Note that cyclically shifting the half-edge sequence does not change $C_v^1$ or $C_v^2$.  Now we prove the three points.  

Fix $v$.  If $w$ is not adjacent to $v$ then we may shift the half-edge sequence so that it is of the form
\begin{equation*}
    \ldots h_v^i h_v^{i+1} \ldots h_w^j h_w^{j+1} \ldots h_w^k h_w^{k+1} \ldots h_v^\ell h_v^{\ell+1} \ldots,
\end{equation*}
where either $h_v^{i+1} =h_v^2$ and $h_v^\ell = h_v^3$ or $h_v^{i+1} =h_v^4$ and $h_v^\ell = h_v^1$.  In the former case $C_v^2$ contains both $w$ transitions while in the later $C_v^1$ contains both, proving point $1$.

To prove point 2, suppose $w$ is adjacent to $v$. Then we may shift the half-edge sequence to be of the form 
\begin{equation*}
    \ldots h_v^i h_v^{i+1} \ldots h_w^j h_w^{j+1} \ldots h_v^k h_v^{k+1} \ldots h_w^\ell h_w^{\ell+1} \ldots.
\end{equation*}
Again either $h_v^{i+1} =h_v^2$ and $h_v^k = h_v^3$ in which case $C_v^2 $ contains only the transition $h_w^j \rightarrow h_w^{j+1}$ or $h_v^{i+1} =h_v^4$ and $h_v^k = h_v^1$ in which case $C_v^1$ contains only $h_w^j \rightarrow h_w^{j+1}$. In either case $C_v^1$ and $C_v^2$ contain one $w$ transition each.  Note also that if $w=v$ then $C_v^1 $ and $C_v^2$ each contain one transition corresponding to the split.  

Point $3$ follows immediately from the previous arguments.
\end{proof}

Some stabilizers of $\Psi$ can be written in terms of the cycles $\{C_v^1, C_v^2 : v \in [n]\}$. Recall that the set of edge operators $\edgeops$ are the stabilizers of $\Psi$ defined in \Cref{def:Psi_hs}.  Define the operators:
\begin{equation}\label{eq:cycle_stabilizer}
    \tilde{S}_v^a \coloneqq \prod_{(h_u^\alpha, h_w^\beta) \in C_v^a} \edgeops(h_u^\alpha, h_w^\beta),
\end{equation}
where $a = 1, 2$ and we abuse notation slightly by writing ``$(h_u^\alpha, h_w^\beta) \in C_v^a$'' to mean the directed edges (not transitions) appearing in the cycle $C_v^a$ and $\edgeops(\cdot{}, \cdot{})$ to mean the stabilizer associated with edge $(h_u^\alpha, h_w^\beta)$. Because each $\edgeops( h_u^\alpha, h_w^\beta)$ stabilizes $\Psi$, so too does any arbitrary product. However, the cycle stabilizers $\tilde{S}_v^a$ are insufficient to generate the stabilizer group of $\Psi$. This can be seen by the fact that cycles are closed under symmetric differences, while \cref{eq:matching_state_defn} contains all subsets of edges of $F$, which includes non-cycles.

We can now derive the CGS stabilizers from the matching state $\Psi$. Recall that graph state stabilizers take the form
\begin{equation}
    S_v = X_v \prod_{w \in \mathcal{N}_G(v)} Z_w.
\end{equation}
In the following theorem, we show that $\tilde{S}_v^1$ and $\tilde{S}_v^2$ are equal up to a sign within the joint eigenspaces of all $D_1, \ldots, D_n$, and furthermore they act as $\pm S_v$. In particular, our choice of convention is such that we get $S_v$ in the joint $+1$-eigenspace.

\begin{theorem}\label{thm:cycles_to_stabilizers}
    Let $F$ be a connected $4$-regular multigraph, $T$ an Eulerian tour on $F$, and $G$ the corresponding circle graph. Let $\tilde{S}_v^a$ be the cycle stabilizers as defined in \cref{eq:cycle_stabilizer}. Let $\subspace{s_1, \ldots, s_n}$ be the $2^n$-dimensional subspace wherein $D_v$ acts as $s_v\I$. Then for all $v \in [n]$, $\tilde{S}_v^1 \cong \pm \tilde{S}_v^2$ inside $\subspace{s_1, \ldots, s_n}$, and we can label the cycles such that $\tilde{S}_v^2 \cong S_v$ within the subspace $\subspace{1, \ldots, 1}$.
\end{theorem}

\begin{proof}
Let us define $C_v^2$ as the cycle taking the transition $h_v^3 \rightarrow h_v^2$ for all $v$.  The edge operators commute and by \Cref{claim:prod_of_gauges} the product of the edge operators is equal to the product of the gauges $\prod_i D_i$.  Hence within any fixed subspace of $\prod_i D_i$ we have that $\tilde{S}_v^1 =\pm \tilde{S}_v^2$.  In the space $\subspace{1, \ldots, 1}$ we have that $\tilde{S}_v^1 = \tilde{S}_v^2$ and we can demonstrate $\tilde{S}_v^1=\tilde{S}_v^2 \cong S_v$ in a subspace by examining only $\tilde{S}_v^2$.  By our convention  $h_v^3$ will always come after $h_v^2$ in the tour and $C_v^2$ will never traverse the first edge of the tour (which has a flipped stabilizer).  If $C_v^2$ has $k$ vertex transitions in addition to the transition at $v$, then WLOG $\tilde{S}_v^2$ has the form
\begin{equation}
\begin{split}
    \tilde{S}_v^2 &=(-i c_v^2 c_{q_1}^{j_1}) (-i c_{q_1}^{j_1+1} c_{q_2}^{j_2})(-ic_{q_2}^{j_2+1} c_{q_3}^{j_3}) \cdots (-i c_{q_{k-1}}^{j_{k-1}+1} c_{q_k}^{j_k})(-ic_{q_k}^{j_k+1} c_v^3)\\
    &=-i c_v^2 \prod_{\ell=1}^k (-i c_{q_\ell}^{j_\ell} c_{q_\ell}^{j_\ell+1}) c_v^3\\
    &=(-i c_v^2 c_v^3)\prod_{\ell=1}^k (-i c_{q_\ell}^{j_\ell} c_{q_\ell}^{j_\ell+1}),
\end{split}
\end{equation}
where in the final line we used the fact that $C_v^2$ visits $v$ exactly once so the transitions $h_{q_\ell}^{j_\ell} \rightarrow h_{q_\ell+1}^{j_\ell+1} $ all occur on different vertices and hence  $c_v^3$ commutes with all $-i c_{q_\ell}^{j_\ell} c_{q_\ell}^{j_\ell+1}$.  Again by our labeling convention, $-i c_{q_\ell}^{j_\ell} c_{q_\ell}^{j_\ell+1}$ will always either be $-i c_{q_\ell}^{1} c_{q_\ell}^{2}$ or $-i c_{q_\ell}^{3} c_{q_\ell}^{4}$ (allowing for the possibility that both terms appear in the product).  

By \Cref{lem:cycle_properties}, if $w$ is adjacent to $v$ in the circle graph, then either $-i c_{w}^{1} c_{w}^{2}$ or $-i c_{w}^{3} c_{w}^{4}$ appear, but not both.  This equates to a $\tilde{Z}_w$ operator.  Also by \Cref{lem:cycle_properties}, if $w$ is not adjacent to $v$, then $C_v^2$ picks up either zero or two instances of $w$. Hence either no operators from $w$ show up (equating to $\I$) or $-i c_{w}^{1} c_{w}^{2}$ and $-i c_{w}^{3} c_{w}^{4}$ both appear (equating to $D_w$).  Note that $(-i c_v^2 c_v^3)=\tilde{X}_v$ always appears in $\tilde{S}_v^2$.  Hence, in the subspace corresponding to $D_w=+1$ for all $w$, we obtain $\tilde{S}_v^2 \cong S_v$.
\end{proof}

\subsection{Explicit representation of encoded circle graph states}\label{sec:circle_graph_encoding}

In \cref{sec:stabilizers_from_cycles} we showed that the stabilizers of the CGS are encoded into the matching state. The following theorem elucidates this encoding at the level of matching state itself, which will inform us how to use $\Psi$ to directly compute properties of $\ket{G}$.

\begin{theorem}\label{thm:embedded_graph_state}
    Let $G$ be a circle graph on $n$ vertices and $\Psi = \op{\Psi}{\Psi}$ be the matching state defined via \Cref{def:Psi_hs}.  Then,
    \begin{equation}
        \ket{\Psi} = \frac{1}{\sqrt{2^{n-1}}} \sum_{\substack{s_1, \ldots, s_n \in \{\pm 1\} \\ \prod_{j \in [n]} s_j = 1}} \ket{\psi(s_1, \ldots, s_n)},
    \end{equation}
    where each $\ket{\psi(s_1, \ldots, s_n)}$ is an eigenstate of all local gauge operators $D_j$ with eigenvalue $s_j$. In particular, we can choose $\ket{\psi(1, \ldots, 1)} \cong \ket{G}$ under the fermion-to-qubit mapping of \cref{claim:pauli_subspace_mapping}.
\end{theorem}

\begin{proof}
    Fix a subspace $\subspace{s_1, \ldots s_n}$. Let $\tilde{S}_j^a$ be the cycle stabilizers of $\ket{\Psi}$ defined in \cref{eq:cycle_stabilizer}. By \cref{thm:cycles_to_stabilizers}, $\tilde{S}_j^1 \cong \pm \tilde{S}_j^2$ within the subspace, leaving us with only $n$ independent generators. Since the dimension of this subspace is $2^n$, there is exactly one stabilizer state per subspace, which we write as $\ket{\psi(s_1, \ldots, s_n)}$. Indeed, we also saw from \cref{thm:cycles_to_stabilizers} that the stabilizers $\tilde{S}_j^a$ act as $S_j$ up to signs, so $\ket{\psi(s_1, \ldots, s_n)}$ are all locally Clifford equivalent to the same CGS, which we can choose to be $\ket{\psi(1, \ldots, 1)}$.

    To see that all states have equal amplitudes in $\ket{\Psi}$ (the phases can be absorbed into the definition of $\ket{\psi(s_1, \ldots, s_n)}$), first note that because $D_1 \cdots D_n \ket{\Psi} = \ket{\Psi}$, it lies within the union of subspaces for which the product of gauges is $+1$. There are only $2^{n-1}$ such subspaces, and we can write
    \begin{equation}
        \ket{\Psi} = \sum_{\substack{s_1, \ldots, s_n \in \{\pm 1\} \\ \prod_{j \in [n]} s_j = 1}} \ip{\psi(s_1, \ldots, s_n)}{\Psi} \ket{\psi(s_1, \ldots, s_n)}.
    \end{equation}
    Because of the global parity constraint, it suffices to apply only $n-1$ of the local projectors to extract one of the subspaces:
    \begin{equation}
       \prod_{j=1}^{n-1} \l( \frac{\I + s_j D_j}{2} \r) \ket{\Psi} = \ip{\psi(s_1, \ldots, s_n)}{\Psi} \ket{\psi(s_1, \ldots, s_n)}.
    \end{equation}
    The squared overlap is therefore
    \begin{equation}
    \begin{split}
        \abs{\ip{\psi(s_1, \ldots, s_n)}{\Psi}}^2 &= \ev{\prod_{j=1}^{n-1} \l( \frac{\I + s_j D_j}{2} \r)}{\Psi}\\
        &= \frac{1}{2^{n-1}} \sum_{A \subseteq [n-1]} \ev{\prod_{j \in A} s_j D_j}{\Psi}.
    \end{split}
    \end{equation}
    Because we assume that $F$ has no self-loops, $\ev{\prod_{j \in A} s_j D_j}{\Psi} = 0$ for every $A \neq \varnothing$. To see why this holds, recall that $\tr(c_Q c_R) = 0$ if and only if $Q \neq R$. Let $Q$ be the indices from $\prod_{j \in A} D_j$, i.e., $Q = \{4j-3, 4j-2, 4j-1, 4j : j \in A\}$, and let $R$ be the set of vertices incident on any subset of edges (i.e., expand $\Psi = \sum_R \Psi_R c_R$). Since $F$ has no self-loops, every half-edge on a vertex $j$ must connect to a half-edge on a different vertex $k$. So there is no $R$ that contains all half-edges from one vertex without also containing some partial collection of half-edges from another, other than $R = \varnothing$ and $R = [4n]$. Because $A \subseteq [n-1]$, we cannot have $Q = [4n]$, so the only nonzero term in the sum is when $Q = R = \varnothing$. Thus
    \begin{equation}
        \abs{\ip{\psi(s_1, \ldots, s_n)}{\Psi}}^2 = \frac{1}{2^{n-1}},
    \end{equation}
    as claimed.
\end{proof}

\subsection{All illustrative example}

\begin{figure}[t]
  \centering
  \begin{subfigure}[b]{0.33\textwidth}
    \centering
    \includegraphics[width=\textwidth]{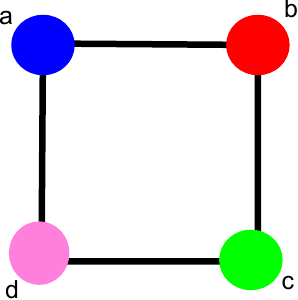}
    \caption{Graph $G$ corresponding to graph state with stabilizer group $\langle X_a Z_b Z_d, X_b Z_a Z_c, X_c Z_b Z_d, X_d Z_a Z_c \rangle$.}
    \label{fig:explicit_G}
  \end{subfigure}
  \hspace{2.5cm}
  \begin{subfigure}[b]{0.33\textwidth}
    \centering
    \includegraphics[width=\textwidth]{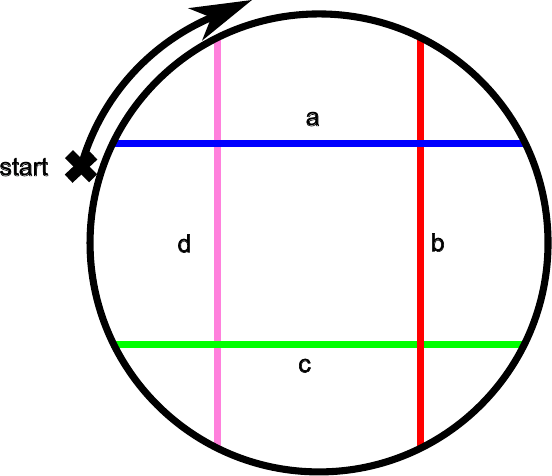}
    \caption{Chord diagram for graph $G$.  Eulerian tour $T$ corresponds to traversing the circle and ``recording'' chords which are seen: $\VS(T)=adbacbdc$.}
    \label{fig:explicit_chord_diagram}
  \end{subfigure}

  \vspace{1em}  

  \begin{subfigure}[b]{0.33\textwidth}
    \centering
    \includegraphics[width=\textwidth]{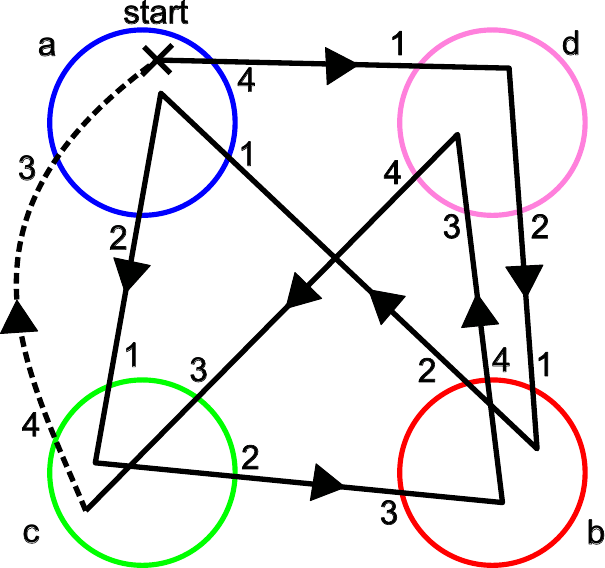}
    \caption{Eulerian tour traced through the vertices with start position matching \Cref{fig:explicit_chord_diagram}.  Half-edges are labeled according to \Cref{def:half_edge_labels}.}
    \label{fig:explicit_tour}
  \end{subfigure}
  \hspace{2.5cm}
  \begin{subfigure}[b]{0.33\textwidth}
    \centering
    \includegraphics[width=\textwidth]{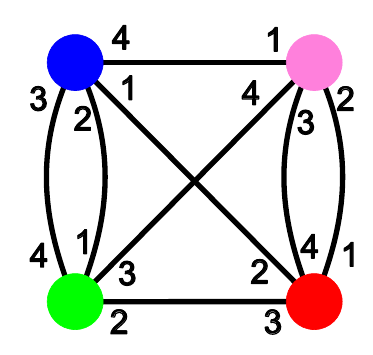}
    \caption{$4$-regular multi-graph $F$ for this example.  Half-edges are labeled according to tour $T$ from \Cref{fig:explicit_tour} and \Cref{def:half_edge_labels}.}
    \label{fig:explicit_F}
  \end{subfigure}

  \caption{Relevant graphs for construction of fermionic matching state corresponding to $C_4$.}
  \label{fig:four_subfigures}
\end{figure}

In this section we detail a small example, explicitly showing how to convert a circle graph state into a fermionic matching state.  The initial graph state here corresponds to the $4$-cycle $C_4$ (see \Cref{fig:explicit_G}).  Graph state stabilizers correspond to Pauli operators with $X$ at vertex $v$ and $Z$ on all neighbors, $\mathcal{N}_G(v)$.  Hence for this case the stabilizer group is $\langle X_a Z_b Z_d, X_b Z_a Z_c, X_c Z_b Z_d, X_d Z_a Z_c \rangle$.  This graph is a circle graph, and has a simple chord diagram given in \Cref{fig:explicit_chord_diagram}.  Vertices of the graph $G$ correspond to the labels of the chords. An Eulerian tour $T$ can be found by fixing some arbitrary starting location and traversing the circle:~$\VS(T)=adbacbdc$.  The order of the chord labels determines the vertex sequence of the tour $T$.  Note that, for instance, vertex $a$ is connected to $b$ and $d$ given the vertex sequence.  

In \Cref{fig:explicit_tour} we depict the Eulerian tour $T$ from \Cref{fig:explicit_chord_diagram} along with arrows indicating the direction of travel.  Half-edges are labeled using \Cref{def:half_edge_labels}.  The half-edge sequence of this tour is given by $\HS(T)=h_a^4 h_d^1 h_d^2 h_b^1 h_b^2 h_a^1 h_a^2 h_c^1 h_c^2 h_b^3 h_b^4 h_d^3 h_d^4 h_c^3 h_c^4 h_a^3$.  

The matching state $\Psi$ in this example has stabilizer group 
\begin{equation}
    \mathcal{S}=\langle i c_a^4 c_d^1, -ic_d^2 c_b^1, -ic_b^2 c_a^1, -ic_a^2 c_c^1, -i c_c^2 c_b^3, -i c_b^4 c_d^3, -i c_d^4 c_c^3, -i c_c^4 c_a^3 \rangle.
\end{equation}
When we multiply all the stabilizers together, we get 
\begin{align}
    \nonumber A&=(i c_a^4 c_d^1)(-ic_d^2 c_b^1)(-ic_b^2 c_a^1)(-ic_a^2 c_c^1)(-i c_c^2 c_b^3)(-i c_b^4 c_d^3)(-i c_d^4 c_c^3)(-i c_c^4 c_a^3)\\
    &=- c_a^4 (c_d^1 c_d^2) (c_b^1 c_b^2) (c_a^1 c_a^2) (c_c^1 c_c^2) (c_b^3 c_b^4) (c_d^3 c_d^4) (c_c^3 c_c^4) c_a^3\\
    \nonumber &= - c_a^4 c_a^1 c_a^2  c_a^3 (c_b^1 c_b^2 c_b^3 c_b^4) (c_c^1 c_c^2 c_c^3 c_c^4) (c_d^1 c_d^2 c_d^3 c_d^4)= D_a D_b D_c D_d.
\end{align}
We can compute the splits as 
\begin{align}
    C_a^2&=h_a^2 h_c^1 h_c^2 h_b^3 h_b^4 h_d^3 h_d^4 h_c^3 h_c^4 h_a^3,\\
    C_b^2&=h_b^2 h_a^1 h_a^2 h_c^1 h_c^2 h_b^3,\\
    C_c^2&=h_c^2 h_b^3 h_b^4 h_d^3 h_d^4 h_c^3,\\
    \text{and } C_d^2&=h_d^2 h_b^1 h_b^2 h_a^1 h_a^2 h_c^1 h_c^2 h_b^3 h_b^4 h_d^3.
\end{align}
Replacing half-edges $h$ with Majorana operators $c$, these correspond to operators 
\begin{align}
\tilde{S}_a^2&=(-i c_a^2 c_c^1)( -i c_c^2 c_b^3) (-i c_b^4 c_d^3) (-i c_d^4 c_c^3) (-i c_c^4 c_a^3), \\
\tilde{S}_b^2&=(-i c_b^2 c_a^1) (-i c_a^2 c_c^1) (-i c_c^2 c_b^3),\\
\tilde{S}_c^2&=(-i c_c^2 c_b^3) (-i c_b^4 c_d^3) (-i c_d^4 c_c^3),\\
\text{and  } \tilde{S}_d^2&=(-i c_d^2 c_b^1) (-i c_b^2 c_a^1) (-i c_a^2 c_c^1) (-i  c_c^2 c_b^3) (-i c_b^4 c_d^3).
\end{align}
These operators simplify to 
\begin{align}
    \tilde{S}_a^2&=-i c_a^2 (-i c_c^1 c_c^2) (-i c_b^3 c_b^4) (-i c_d^3 c_d^4) (-i c_c^3 c_c^4) c_a^3= X_a Z_b Z_d D_c,\\
    \tilde{S}_b^2&=-i c_b^2 (-i c_a^1 c_a^2) (-i c_c^1 c_c^2) c_b^3=X_b Z_a Z_c,\\
    \tilde{S}_c^2&=-i c_c^2 (-i c_b^3 c_b^4) (-i c_d^3 c_d^4) c_c^3=X_c Z_b Z_d,\\
    \text{and } \tilde{S}_d^2&=-i c_d^2 (-i c_b^1 c_b^2)  (-i c_a^1 c_a^2) (-i c_c^1 c_c^2) (-i c_b^3 c_b^4) c_d^3=X_d D_b Z_a Z_c.
\end{align}
Using the fact that $D_j \cong \I$ on the appropriate subspace, we see that we have recovered the desired stabilizers.

\section{Fermionic Representation of Product States}\label{sec:product_states}

In this section, we describe how product states are encoded into the fermionic space. Let $\ket{\phi}$ be a single-qubit pure state. Denote its Bloch vector coordinates as $x = \ev{X}{\phi}, y = \ev{Y}{\phi}, z = \ev{Z}{\phi}$. Taking the equivalences of \cref{eq:qubit_mapping,eq:pauli_equiv,eq:qubit_mapping_aux} with gauge value $s = 1$, we are motivated to define the fermionic two-mode density operator $\tilde{\phi}$ whose covariance matrix is
\begin{equation}\label{eq:single-qubit_covmat}
    \Gamma_{\phi} = \begin{pmatrix}
        0 & z & -y & x\\
        -z & 0 & x & y\\
        y & -x & 0 & z\\
        -x & -y & -z & 0
    \end{pmatrix}
\end{equation}
(to ease notation, we drop tildes on these covariance matrices). By construction, $\tilde{\phi}$ recovers the correct Pauli expectations under the fermion-to-qubit mapping. Furthermore the columns of $\Gamma_\phi$ are orthonormal, so by \cref{fact:cov_mat_decomp} we know it is Gaussian. Observe that $\pf(\Gamma_{\phi}) = x^2 + y^2 + z^2 = 1$, because $\ket{\phi}$ is pure. Wick's theorem then tells us that $\pf(\Gamma_{\phi}) = \ev{(-c^1 c^2 c^3 c^4)}{\tilde{\phi}}$, which is the expectation of the gauge operator $D$. Thus $\ket{\tilde{\phi}}$ is a $+1$-eigenstate of $D$, as expected by our choice of taking $s = 1$ at the beginning.

An $n$-qubit product state $\ket{\Phi} = \bigotimes_{j \in [n]} \ket{\phi_j}$ can then be mapped to the $2n$-mode Gaussian state $\ket{\tilde{\Phi}}$ with covariance matrix
\begin{equation}\label{eq:blockdiag_covmat}
    \Gamma_\Phi \coloneqq \bigoplus_{j \in [n]} \Gamma_{\phi_j}.
\end{equation}
Our construction naturally ensures that the projection of $\ket{\tilde{\Phi}}$ to $\subspace{1, \ldots, 1}$ is precisely $\ket{\Phi}$ (up to some unimportant global phase).

\subsection{Product unitaries and matchgates}\label{sec:matchgates}

Single-qubit unitaries correspond to Gaussian unitaries under the gauge $D=+1$.  First note that under $D=+1$, $X \cong -i c^1 c^4$ and $Z \cong -i c^3 c^4$ are quadratic in the Majorana operators $\{c^1, c^2, c^3, c^4\}$.  Hence unitaries of the form $e^{i \theta X}$ and $e^{i \theta Z}$ for $\theta \in \mathbb{R}$ are Gaussian unitaries~\cite{bravyi2004lagrangian}.  Since rotations in $X$ and $Z$ generate $\SU(2)$, it must be that arbitrary single-qubit unitaries can be expressed as Gaussian unitaries under the $+1$ gauge.  We note that $e^{i \theta X}$ and $e^{i\theta Z}$ commute with $D$ and hence the action of Gaussian unitaries corresponding to single-qubit Pauli rotations maintain the gauge degree of freedom.  

Although we will not need it in our proofs of MBQC simulability, it is worth describing these unitaries explicitly. Let
\begin{equation}
    U = \begin{pmatrix}
        U_{00} & U_{01}\\
        U_{10} & U_{11}
    \end{pmatrix}
\end{equation}
be a unitary such that $\ket{\phi} = U \ket{0}$. Lifting this to the fermionic space, one obtains the matchgate
\begin{equation}\label{eq:matchgate}
    M = \begin{pmatrix}
        U_{00} & 0 & 0 & U_{01}\\
        0 & U_{00} & U_{01} & 0\\
        0 & U_{10} & U_{11} & 0\\
        U_{10} & 0 & 0 & U_{11}
    \end{pmatrix}.
\end{equation}
This can be seen by the fact that in the $+1$-gauge specified by \cref{eq:single-qubit_covmat}, we have the mapping from qubit to fermion states as $\ket{0}_{\mathrm{q}} \mapsto \ket{00}_{\mathrm{f}}$, $\ket{1}_{\mathrm{q}} \mapsto \ket{11}_{\mathrm{f}}$. In turn, the matchgate $M$ possesses an $\SO(4)$-representation~\cite{knill2001fermionic,jozsa2008matchgates} such that
\begin{equation}
    \Gamma_\phi = O \begin{pmatrix}
        0 & 1 & 0 & 0\\
        -1 & 0 & 0 & 0\\
        0 & 0 & 0 & 1\\
        0 & 0 & -1 & 0
    \end{pmatrix} O^\T.
\end{equation}
The fermionic Gaussian unitary is then the Bogoliubov transformation according $O$, as given by \cref{def:FGU}.

The product unitary over $n$ qubits is then simply the tensor product of $n$ such matchgates, which is represented in the covariance formalism by the $4n \times 4n$ block-diagonal orthogonal matrix $\bigoplus_{j \in [n]} O_j$. Note that this also gives us a way to explicitly write down the fermionic operators $\tilde{\phi}_j$ using \cref{eq:gaussian_state_defn}.

We summarize the main results of this entire section in the following statement.

\begin{lemma}\label{lem:prod_state_fermions}
    Let $\ket{\Phi} = \bigotimes_{j \in [n]} \ket{\phi_j}$ be a product state on $n$ qubits. To each $\ket{\phi_j}$, we define
    \begin{equation}
        \tilde{\phi}_j \coloneqq M_j \l( \frac{\I - i c_j^1 c_j^2}{2} \r) \l( \frac{\I - i c_j^3 c_j^4}{2} \r) M_j^\dagger,
    \end{equation}
    where $M_j$ is the matchgate from $\cref{eq:matchgate}$ acting on Majorana modes $4j - 3, \ldots, 4j$. Then
    \begin{equation}\label{eq:Phitilde_operator}
        \tilde{\Phi} \coloneqq \prod_{j \in [n]} \tilde{\phi}_j
    \end{equation}
    is a pure FGS with covariance matrix given by \cref{eq:blockdiag_covmat}. Furthermore, $\ket{\tilde{\Phi}}$ is a $+1$-eigenstate of all $D_1, \ldots, D_n$ (where $\ket{\tilde{\Phi}}$ is such that $\tilde{\Phi} = \op{\tilde{\Phi}}{\tilde{\Phi}}$).
\end{lemma}

\section{Efficient Classical Simulation of MBQC under Circle Graph States}\label{sec:classical_sim}

In this section, we accumulate the technical results developed above to demonstrate our main result:~how to simulate an MBQC protocol using circle graph resource states in polynomial time.

\subsection{Product-state overlaps}\label{sec:product_overlap}

First we show how to compute the overlap between a CGS and an arbitrary product state on all $n$ qubits. This is the necessary ingredient for executing the gate-by-gate algorithm of \cite{bravyi2022simulate}. Because $\Psi$ and $\tilde{\Phi}$ are both Gaussian, their overlap is efficiently computed via a Pfaffian of their covariance matrices. The formula is given by the following well-known result.

\begin{proposition}[{\cite[Equation 22]{bravyi2017complexity}}]\label{prop:pure_gaussian_overlap}
    Let $\ket{\varphi_1}, \ket{\varphi_2}$ be two $N$-mode FGSs of the same parity $\sigma \in \pmone$. If $\Gamma_{\varphi_1}, \Gamma_{\varphi_2}$ are their covariance matrices, then
    \begin{equation}\label{eq:gaussian_overlap}
        \abs{\ip{\varphi_1}{\varphi_2}}^2 = \frac{\sigma}{2^N} \pf(\Gamma_{\varphi_1} + \Gamma_{\varphi_2}).
    \end{equation}
\end{proposition}

Recall that the Pfaffian can be computed in $\O(N^3)$ time. Equivalently, one can also use the identity $\pf(X)^2 = \det(X)$ for skew-symmetric $X$ and the fact that \cref{eq:gaussian_overlap} is non-negative. Clearly if the states have different parities, then they are automatically orthogonal (one can also verify this by checking that the Pfaffian above vanishes in that case). Note that \cite{bravyi2017complexity} also shows that the phase of the overlap, with respect to some fixed reference state, can be determined in cubic time as well. For our purposes, the magnitude is sufficient.

The overlap that we want is given by the following formula.

\begin{theorem}\label{thm:prod_state_overlap}
    Let $\ket{G}$ be an $n$-qubit CGS and $\ket{\Phi}$ a product state. Then
    \begin{equation}
        \abs{\ip{\Phi}{G}}^2 = \frac{1}{2^{n+1}} \pf(\Gamma_\Phi + \Gamma_\Psi),
    \end{equation}
    where $\Psi$ is the matching state associated with $G$.
\end{theorem}

\begin{proof}
    Let $\tilde{\Phi}$ be the fermionic representation of the product state $\ket{\Phi}$ (\cref{lem:prod_state_fermions}). By our convention, both $\tilde{\Phi}$ and $\Psi$ have parity $\sigma = 1$. Thus by \cref{prop:pure_gaussian_overlap} we have
    \begin{equation}\label{eq:fermion_product_overlap}
        \abs*{\ip{\tilde{\Phi}}{\Psi}}^2 = \frac{1}{2^{2n}} \pf(\Gamma_\Phi + \Gamma_\Psi).
    \end{equation}
    On the other hand, by \cref{thm:embedded_graph_state} we get
    \begin{equation}\label{eq:gauge_overlap}
        \ip{\tilde{\Phi}}{\Psi} = \frac{1}{\sqrt{2^{n-1}}} \ip{\tilde{\Phi}}{\psi(1, \ldots, 1)},
    \end{equation}
    since $\ket{\tilde{\Phi}}$ lies in the joint $+1$-eigenspace of all $D_1, \ldots, D_n$. But since $\ket{\psi(1, \ldots, 1)} \cong \ket{G}$ and $\ket{\tilde{\Phi}} \cong \ket{\Phi}$, we have $\abs*{\ip{\tilde{\Phi}}{\psi(1, \ldots, 1)}} = \abs{\ip{\Phi}{G}}$. Combining \cref{eq:fermion_product_overlap,eq:gauge_overlap} yields the claim.
\end{proof}

\subsection{Weak simulation from strong simulation}

The above result shows that the overlap between CGSs and arbitrary product states can be computed in polynomial time. However, in MBQC the unitary applied to a given qubit may depend on the outcomes of earlier measurements, so this overlap computation alone is not generally sufficient to sample from the output distribution of an MBQC protocol. A natural approach is to use the probability chain rule and sample the measurement outcomes one qubit at a time from the corresponding conditional marginals. Unfortunately, evaluating such marginals can be $\sharpP$-hard \cite{bravyi2022simulate}, and this hardness is known to hold even for circle graphs \cite{hahn2026structure}. Fortunately, the gate-by-gate algorithm of \cite{bravyi2022simulate} allows one to bypass explicit marginal computations via a clever sampling algorithm that requires only global product-state overlap evaluations.

We consider an MBQC protocol in which each qubit is first acted on by a unitary and then measured in the computational basis. Without loss of generality, the qubits are measured sequentially from $1$ to $n$ (otherwise relabel the qubits). The unitary applied to qubit $i$ may depend on previous outcomes, so we write it as $U_i(x_1,\ldots,x_{i-1})$.  For an arbitrary resource state $\ket{\psi}$, define 
\begin{equation}\label{eq:bravyi_partial}
    P_t(x) \coloneqq \l|\langle x|U_1 \otimes U_2(x_2)\otimes ...\otimes U_t(x_1, x_2, ..., x_{t-1}) \otimes \mathbb{I}_{2^{n-t}} |\psi \rangle \r{|^2}.
\end{equation}
The goal is to construct a polynomial-time randomized algorithm which samples a bit-string $x$ with probability to $P_n(x)$.

\begin{proposition}[{\cite[Algorithm 3]{bravyi2022simulate}}]\label{thm:no_marg}
    Given a randomized algorithm that samples $x$ according to $P_0(x)=|\langle x | \psi\rangle|^2$, together with a polynomial-time procedure for computing $P_t(x)$ for all $x\in\{0,1\}^n$ and $t\in[n]$, one can construct a polynomial-time randomized algorithm that samples $x$ according to $P_n(x)$.
\end{proposition}

As a consequence, we can efficiently simulate MBQC with circle graph resource states.

\begin{corollary}
    Let $\ket{G}$ be an $n$-qubit circle graph state. There exists a randomized polynomial-time algorithm that samples $x\in\{0,1\}^n$ according to $P_n(x)$, where $P_n(x)$ denotes the probability of obtaining outcome string $x$ at the end of an MBQC protocol that uses $\ket{G}$ as its resource state.   
\end{corollary}
\begin{proof}
    It is clear that $P_t(x)$ may be computed for all $x$ and $t$ in $\O(n^3)$ time using \Cref{thm:prod_state_overlap}.  In order to use \Cref{thm:no_marg} we must also demonstrate a randomized algorithm for sampling $x$ with probability $|\langle x|G\rangle|^2$. But this is easily accomplished in $\O(n^3)$ time by the Gottesman--Knill theorem~\cite{AaronsonGottesman2004}, since $\ket{G}$ is a stabilizer state.
\end{proof}

\section{Acknowledgments}

We thank Adrian Chapman, Nathan Claudet, David Gosset, and Robert Raussendorf for insightful discussions. We also thank an anonymous reviewer for bringing to our attention the algorithm of \cite{bravyi2022simulate}. This material is based upon work supported by the U.S.\ Department of Energy, Office of Science, Accelerated Research in Quantum Computing, Fundamental Algorithmic Research toward Quantum Utility (FAR-Qu). BH is supported by the US NSF grants PHYS-1820747 and NSF (EPSCoR-1921199) as well as the Office of Science, Office of Advanced Scientific Computing Research under program Fundamental Algorithmic Research for Quantum Computing. BH is supported by the ``Quantum Chemistry for Quantum Computers'' project sponsored by the DOE, Award DE-SC0019374. This work was completed in part while BH was a visiting scholar at the Simons Institute for the Theory of Computing. VI is supported by an NSF Graduate Research Fellowship. AZ is supported by the Laboratory Directed Research and Development program at Sandia National Laboratories, under the Gil Herrera Fellowship in Quantum Information Science.

This article has been authored by an employee of National Technology \& Engineering Solutions of Sandia, LLC under Contract No.\ DE-NA0003525 with the U.S. Department of Energy (DOE). The employee owns all right, title and interest in and to the article and is solely responsible for its contents. The United States Government retains and the publisher, by accepting the article for publication, acknowledges that the United States Government retains a non-exclusive, paid-up, irrevocable, world-wide license to publish or reproduce the published form of this article or allow others to do so, for United States Government purposes. The DOE will provide public access to these results of federally sponsored research in accordance with the DOE Public Access Plan \url{https://www.energy.gov/downloads/doe-public-access-plan}.

\printbibliography
\addcontentsline{toc}{section}{References}

\appendix

\section{Definition and Properties of Entanglement Width}\label{sec:EntanglementWidth}

A necessary condition for a family of states to be universal for MBQC is that it exhibit unbounded entanglement width \cite{van_den_nest_universal_2006}. In this appendix we will define this concept, starting with some prerequisites.

\subsection{Schmidt rank}

Let $\mathcal{H}_A$ and $\mathcal{H}_B$ be Hilbert spaces of dimensions $n$ and $m$ respectively. WLOG, we may take $n \geq m$. Then for any state $\ket{\psi} \in \mathcal{H}_A \otimes \mathcal{H}_B$, there exist orthonormal sets $\{\ket{u_1}, \ket{u_2}, \ldots, \ket{u_m}\} \in \mathcal{H}_A$, $\{\ket{v_1}, \ket{v_2}, \ldots, \ket{v_m}\} \in \mathcal{H}_B$ such that
\begin{equation}
\ket{\psi} = \sum_{i=1}^{m}\alpha_i\ket{u_i}\ket{v_i},
\end{equation}
where the $\alpha_i$ are real, non-negative and unique up to reordering. This is the well-known Schmidt decomposition.

The \emph{Schmidt rank} $\chi_{AB}$ of a state with respect to a bipartition of its Hilbert space is the smallest possible number of non-zero coefficients in the associated Schmidt decomposition. The state is entangled with respect to this bipartition if its Schmidt rank is greater than 1; otherwise it is a separable state.

\subsection{Cut rank}

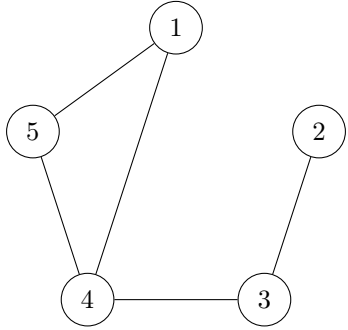
\begin{figure}[H]

\centering

\begin{tikzpicture}[scale=2, every node/.style={circle, draw, minimum size=7mm}]
  \node (1) at (90:1) {1};
  \node (2) at (18:1) {2};
  \node (3) at (-54:1) {3};
  \node (4) at (-126:1) {4};
  \node (5) at (162:1) {5};

  \draw (1) -- (4);
  \draw (2) -- (3);
  \draw (3) -- (4);
  \draw (4) -- (5);
  \draw (5) -- (1);
\end{tikzpicture}

\caption{An example 5-node graph}
\label{fig:5_node_graph}

\end{figure}

Consider a graph $G = (V, E)$ with adjacency matrix $\Gamma$. Then for vertex sets $A, B \subseteq V$, we will define $\Gamma[A,B]$ to be the $\abs{A} \times \abs{B}$ matrix that describes the connectivity between sets $A$ and $B$. We associate each row of $\Gamma[A,B]$ with an element $a \in A$ and each column with an element $b \in B$. Then

\begin{equation}
\Gamma[A,B]_{ab} = \begin{cases}
1, &(a,b)\in E, \\
0, &\text{otherwise}.
\end{cases}
\end{equation}
As an example, consider the graph in \cref{fig:5_node_graph}. Let $A = \{1, 2\}$ and $B = \{3,4,5\}$. Then 

\begin{equation}\label{eq:GammaAB}
\Gamma[A,B] = \begin{pmatrix}
0 & 1 & 1 \\
1 & 0 & 0
\end{pmatrix}.
\end{equation}

We can now define the cut-rank function with respect to a bipartition of V into sets $A$, $V\setminus A$.
\begin{equation}
\text{cutrk}_A(G) \equiv \text{rank}_{\mathbb{F}_2}\Gamma[A, V \setminus A].
\end{equation}
For the example in \cref{eq:GammaAB}, the cut rank would be 2.

In Proposition 10 of \cite{hein2006entanglement} it is shown that the Schmidt rank of a graph state bipartition is equivalent to the cut rank of the corresponding bipartition of its associated graph. 

\subsection{Rank width}\label{sec:rank_width}

We will now discuss the concept of rank width.

A \emph{subcubic tree} is a connected acyclic graph whose vertices have maximum degree 3. See \cref{fig:subcubic} for an example. Notice that deleting any edge from such a tree disconnects it. In particular, it induces a bipartition of the leaves of the tree.

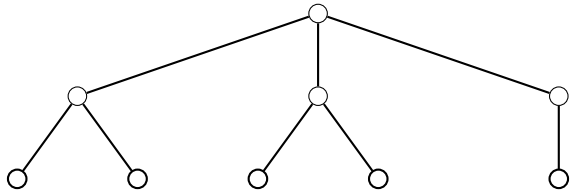
\begin{figure}[H]
\centering    
\begin{tikzpicture}[
  every node/.style={circle, draw, inner sep=1.2pt, minimum size=7pt},
  edge from parent/.style={draw, thick},
  level distance=11mm,
  level 1/.style={sibling distance=32mm},
  level 2/.style={sibling distance=16mm},
  level 3/.style={sibling distance=10mm}
]

\node (r) {} 
  child { node {}
    child { node {} }
    child { node {} }
  }
  child { node {}
    child { node {} }
    child { node {} }
  }
  child { node {}
    child { node {} }
  };

\end{tikzpicture}
\caption{An example of a subcubic tree graph.}
\label{fig:subcubic}
\end{figure}

Let us now consider bipartitions of the vertices of a graph $G = (V,E)$. Specifically, we will associate these vertices one-to-one with the leaves of some subcubic tree and consider the bipartitions induced by removing individual edges of the tree. Note that any bipartition of the vertices can be represented by an edge of \emph{some} subcubic tree, but the converse does not hold. In other words, a given subcubic tree does not in general have an edge corresponding to every bipartition of the associated vertices.

Formally, a \emph{rank decomposition} of a graph G is a pair $R = (T, \mu)$, where $T$ is a subcubic tree and $\mu$ is a bijection $\mu: V(G) \rightarrow \{l : l \text{ is a leaf of } T \}$. Removing any edge $e$ of $T$ induces a bipartition of its leaves and hence of $V(G)$. We say that the \emph{width} of the edge $e$ is equal to the cut rank of this bipartition, which we denote $\mathrm{cutrk}_e(G)$.

In \cref{fig:subcubic2}, we give an example of a rank decomposition of a 5-node graph. Notice that deleting the dashed edge induces a bipartition of the vertices of the graph into sets $\{1,2\}$ and $\{3,4,5\}$. Suppose that the 5-node graph associated with this rank decomposition is the graph in \cref{fig:5_node_graph}. Then the width of the dashed edge would be 2.

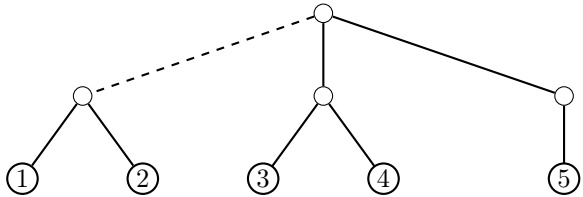
\begin{figure}[H]
\centering    
\begin{tikzpicture}[
  every node/.style={circle, draw, inner sep=1.2pt, minimum size=7pt},
  edge from parent/.style={draw, thick},
  level distance=11mm,
  level 1/.style={sibling distance=32mm},
  level 2/.style={sibling distance=16mm},
  level 3/.style={sibling distance=10mm}
]

\node (r) {} 
  child { node {}
    child { node[thick, solid] {1} [thick, solid]}
    child { node [thick, solid] {2} [thick, solid]}
    edge from parent[draw, thick, dashed]
  }
  child { node {}
    child { node {3} }
    child { node {4} }
  }
  child { node {}
    child { node {5} }
  };

\end{tikzpicture}
\caption{An example of a rank-decomposition of a 5-node graph.}
\label{fig:subcubic2}
\end{figure}

We can now define rank width as the maximum width of the rank decomposition of a graph, minimized with respect to all possible rank decompositions:
\begin{equation}
\mathrm{rwd}(G) \equiv \min\limits_{T} \max_{e\in E(T)}\mathrm{cutrk}_e(G),
\end{equation}

\subsection{Entanglement width}

We can now define the entanglement width of a quantum state $\ket{\psi}$ by a similar procedure. Again consider a subcubic tree $T$, and associate each of its leaves with one of the qubits of state $\ket{\psi}$. Any edge $e$ of $T$ induces a bipartition of these qubits, with an associated Schmidt rank $\chi_e(\ket{\psi})$. The entanglement width of a state $\ket{\psi}$ is then given by 

\begin{equation}
\mathrm{ewd}(\ket{\psi}) \equiv \min\limits_T \max_{e\in E(T)} \chi_e(\ket{\psi}).
\end{equation}
As claimed, the entanglement width of the graph state $\ket{G}$ is equivalent to the rank width of the graph $G$.

\section{Vertex Minors of Graphs}

\subsection{Local complementation}

Local complementation 
on the vertex $v$ of a graph $G$ is an operation that replaces the subgraph induced by the neighborhood of $v$ by its complement. We denote the resultant graph $G * v$. See Fig.~\ref{fig:localComplementation} for an example.

In other words, two vertices $x$ and $y$ are adjacent in $G*v$ if and only if \emph{exactly} one of the following holds~\cite{hein2006entanglement}:

\begin{itemize}
\item $x$ and $y$ are adjacent in $G$, or
\item Both $x$ and $y$ are neighbors of $v$ in $G$.
\end{itemize}

Remarkably, two graph states $\ket{G_1}, \ket{G_2}$ are Local Clifford equivalent if and only if the graphs $G_1$ and $G_2$ are equivalent up to some sequence of local complementations~\cite{hein2006entanglement}.

\begin{definition}
A graph $H$ is said to be a \emph{vertex minor} of a graph $G$ if it can be obtained from $G$ by some sequence of local complementation and vertex deletion operations.
    
\end{definition}

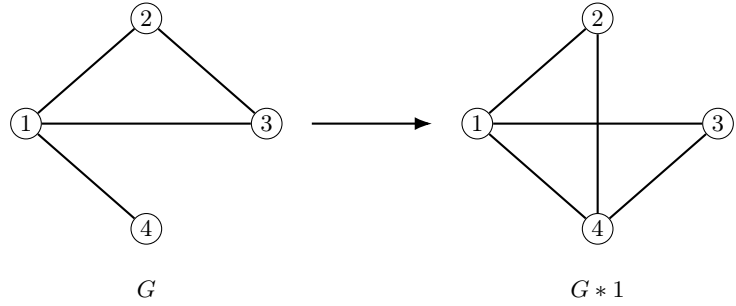
\begin{figure}
    \centering
    \begin{tikzpicture}[every node/.style={circle,draw,inner sep=1.5pt,font=\small}, >=Latex]

\node (A1) at (-1.6,0) {1};
\node (A2) at (0,1.4) {2};
\node (A3) at (1.6,0) {3};
\node (A4) at (0,-1.4) {4};

\draw[thick] (A1)--(A2);
\draw[thick] (A1)--(A3);
\draw[thick] (A1)--(A4);
\draw[thick] (A2)--(A3);

\node[draw=none] at (0,-2.2) {$G$};

\draw[->,thick] (2.2,0) -- (3.8,0);

\node (B1) at (4.4,0) {1};
\node (B2) at (6,1.4) {2};
\node (B3) at (7.6,0) {3};
\node (B4) at (6,-1.4) {4};

\draw[thick] (B1)--(B2);
\draw[thick] (B1)--(B3);
\draw[thick] (B1)--(B4);
\draw[thick] (B2)--(B4);
\draw[thick] (B3)--(B4);

\node[draw=none] at (6,-2.2) {$G * 1$};

\end{tikzpicture}
    \vspace{-10mm}
    \caption{An example of local complementation. On the left, we have a graph $G$, and on the right the graph $G * 1$, which is the result of local complementation with respect to vertex 1 of graph $G$.}
    \label{fig:localComplementation}
\end{figure}

\section{Comparability Grids}\label{sec:comp_grid}

The $n \times n$ \emph{comparability grid} is a graph with vertices labeled by ordered pairs $V = \{(i,j)\mid i,j \in \{1,2,\ldots,n\}\}$. It has an edge between vertices $(i,j)$ and $(i',j')$ if one of the following holds:
\begin{itemize}
\item $i \leq i'$ and $j \leq j'$, or
\item $i \geq i'$ and $j \geq j'$. 
\end{itemize}
See Fig.~\ref{fig:comparability_grid} for an example.
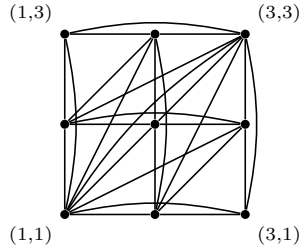
\begin{figure}[H]
    \centering
    \begin{tikzpicture}[x=1.2cm,y=1.2cm,
  vertex/.style={circle,fill,inner sep=1.2pt},
  edge/.style={line width=0.6pt}
]

\foreach \i in {1,2,3}{
  \foreach \j in {1,2,3}{
    \node[vertex] (v-\i-\j) at (\i,\j) {};
  }
}

\node[below left=1pt]  at (v-1-1) {$\scriptstyle(1,1)$};
\node[below right=1pt] at (v-3-1) {$\scriptstyle(3,1)$};
\node[above left=1pt]  at (v-1-3) {$\scriptstyle(1,3)$};
\node[above right=1pt] at (v-3-3) {$\scriptstyle(3,3)$};

\foreach \i in {1,2,3}{
  \foreach \j in {1,2,3}{
    \foreach \ip in {\i,...,3}{
      \foreach \jp in {\j,...,3}{
        \ifnum\i=\ip\relax
          \ifnum\j=\jp\relax
          \else
            \pgfmathtruncatemacro{\dj}{\jp-\j}
            \ifnum\dj=2\relax
              \draw[edge] (v-\i-\j) to[bend right=12] (v-\ip-\jp);
            \else
              \draw[edge] (v-\i-\j)--(v-\ip-\jp);
            \fi
          \fi
        \else
          \pgfmathtruncatemacro{\di}{\ip-\i}
          \pgfmathtruncatemacro{\dj}{\jp-\j}
          \ifnum\di=2\relax
            \ifnum\dj=0\relax
              \draw[edge] (v-\i-\j) to[bend left=12] (v-\ip-\jp);
            \else
              \ifnum\dj=2\relax
                \draw[edge] (v-\i-\j) to[bend left=12] (v-\ip-\jp);
              \else
                \draw[edge] (v-\i-\j)--(v-\ip-\jp);
              \fi
            \fi
          \else
            \draw[edge] (v-\i-\j)--(v-\ip-\jp);
          \fi
        \fi
}}}}

\end{tikzpicture}
    \vspace{-10mm}
    \caption{The $3 \times 3$ comparability grid}
    \label{fig:comparability_grid}
\end{figure}
The comparability grid is itself a circle graph, with the useful property of having all (sufficiently smaller) circle graphs as vertex minors~\cite{geelen_grid_2023}.

\begin{fact}\label{fact:comparabilityGrid}
    Every circle graph on $n$ vertices is isomorphic to a vertex minor of the $3n \times 3n$ comparability grid.
\end{fact}

Recall that we are interested in families of graphs which contain all circle graphs as vertex minors. Due to Fact~\ref{fact:comparabilityGrid}, it is equivalent to consider families of graphs which contain all $n \times n$ comparability grids as vertex minors. The minimal example of such a family would be the family of $n \times n$ comparability grids itself. Indeed, the statement that any family of graphs with unbounded rank width must contain all circle graphs as vertex minors can be phrased with reference to comparability grids.

\begin{theorem}[{\cite[Theorem 1.3]{geelen_grid_2023}}]
    There exists a function $f: \mathbb{Z}\rightarrow\mathbb{Z}$ such that for all positive integers $n$, every graph of rank width at least $f(n)$ has a vertex minor isomorphic to the $n \times n$ comparability grid. 
\end{theorem}

\end{document}